\documentclass[10pt]{article}

\usepackage[utf8]{inputenc}
\usepackage[a4paper]{geometry}
\usepackage{amsmath}  
\usepackage{amssymb}
\usepackage{xcolor}
\usepackage{xspace}
\usepackage{url}
\usepackage{subfigure}

\makeatletter
\let\proof\@undefined          
\let\endproof\@undefined       
\makeatother

\usepackage[noconfig]{ntheorem}
\newtheorem{theorem}{Theorem}[section]
\newtheorem{lemma}[theorem]{Lemma}
\newtheorem{corollary}[theorem]{Corollary}

\theorembodyfont{\normalfont}
\newtheorem{example}{Example}
\newtheorem{definition}{Definition}

\theoremheaderfont{\sc}
\theorembodyfont{\itshape}
\newtheorem{claim}{Claim}

\theoremheaderfont{\scshape}
\theorembodyfont{\normalfont}
\newtheorem*{proof}{Proof}

\usepackage{tikz}
\usepackage{colonequals}
\usepackage{microtype}

\def\qed{{                 
   \parfillskip=0pt        
   \widowpenalty=10000     
   \displaywidowpenalty=10000  
   \finalhyphendemerits=0  
   \leavevmode             
   \unskip                 
   \nobreak                
   \hfil                   
   \penalty50              
   \hskip.2em              
   \null                   
   \hfill                  
   $\square$
   \par}}                  

\newcommand{\learner}{\ensuremath{\mathsf{learner}}}
\newcommand{\SelPath}{\ensuremath{\mathit{SelPath}}}
\newcommand{\Paths}{\ensuremath{\mathit{Paths}}}
\newcommand{\Fusions}{\ensuremath{\mathit{Fusions}}}

\newcommand{\dblslash}{\ensuremath{/\!/}}
\newcommand{\wc}{\mathord{\star}}
\newcommand{\sep}{\mathord{\mathord{:}\mathord{:}}}

\newcommand{\lex}{\mathit{lex}}
\newcommand{\can}{\mathit{can}}

\newcommand{\lab}{\mathit{lab}}
\renewcommand{\root}{\mathit{root}}
\newcommand{\sel}{\mathit{sel}}

\newcommand{\child}{\mathit{child}}
\newcommand{\desc}{\mathit{desc}}

\newcommand{\origin}{\mathit{origin}}
\newcommand{\true}{\mathit{true}}
\newcommand{\false}{\mathit{false}}

\newcommand{\Tree}{\ensuremath{\mathit{Tree}}}
\newcommand{\Twig}{\ensuremath{\mathit{Twig}}}
\newcommand{\Path}{\ensuremath{\mathit{Path}}}
\newcommand{\APath}{\ensuremath{\mathit{AnchPath}}}
\newcommand{\CPath}{\ensuremath{\mathit{ConjPath}}}
\newcommand{\PTwig}{\mathit{PsfTwig}}

\newcommand{\PATH}{{\ensuremath{\mathsf{Path}}}}

\newcommand{\TWIG}{{\ensuremath{\mathsf{Twig}}}}

\newcommand{\APATH}{{\ensuremath{\mathsf{AnchPath}}}}
\newcommand{\PTWIG}{\mathsf{PsfTwig}}
\newcommand{\CPATH}{\ensuremath{\mathsf{ConjPath}}}

\newcommand{\poly}{\mathit{poly}}
\newcommand{\CS}{\mathit{C\!S}}
\newcommand{\Null}{\ensuremath{\text{\sc{}Null}}}
\newcommand{\minus}{\mathbin\setminus}

\newcommand{\Prop}[1]{\ensuremath{\mathbf{P}_{#1}}}

\newcounter{lineNoCounter}
\newcommand{\resetLineNoCounter}{\setcounter{lineNoCounter}{0}}
\newcommand{\lineNo}{\makebox[0pt][l]{%
\makebox[0pt][l]{%
\addtocounter{lineNoCounter}{1}%
}%
\makebox[10pt][r]{
\sf\scriptsize%
\arabic{lineNoCounter}:%
}%
}}

\title{Learning XML Twig Queries\\[-5pt]
\textbf{\small Accepted to ICDT 2012}}
\author{
Sławek Staworko\\
Mostrare, INRIA \& LIFL (CNRS UMR8022)\\
University of Lille, France\\
\texttt{slawomir.staworko@inria.fr}
\and
Piotr Wieczorek\\
Institute of Computer Science\\
University of Wrocław\\ 
\texttt{piotr.wieczorek@cs.uni.wroc.pl}
}
\begin{document}
\maketitle
\thispagestyle{empty}
\begin{abstract}
  We investigate the problem of learning XML queries, \emph{path}
  queries and \emph{twig} queries, from examples given by the user. A
  learning algorithm takes on the input a set of XML documents with
  nodes annotated by the user and returns a query that selects the
  nodes in a manner consistent with the annotation. We study two
  learning settings that differ with the types of annotations. In the
  first setting the user may only indicate \emph{required nodes} that
  the query must select (i.e., \emph{positive examples}). In the
  second, more general, setting, the user may also indicate
  \emph{forbidden nodes} that the query must not select (i.e.,
  \emph{negative examples}). The query may or may not select any node
  with no annotation.

  We formalize what it means for a class of queries to be
  \emph{learnable}. One requirement is the existence of a learning
  algorithm that is \emph{sound} i.e., always returning a query
  consistent with the examples given by the user. Furthermore, the
  learning algorithm should be \emph{complete} i.e., able to produce
  every query with sufficiently rich examples. Other requirements
  involve tractability of the learning algorithm and its robustness to
  nonessential examples. We identify practical classes of Boolean and
  unary, path and twig queries that are learnable from positive
  examples. We also show that adding negative examples to the picture
  renders learning unfeasible.
\end{abstract}
\section{Introduction}
XML has become a de facto standard for representation and exchange of
data in web applications. An XML document is basically a labeled tree
whose leaves store textual data and the standard XML format is text
based to allow users an easy and direct access to the contents of the
document~\cite{XML10}. However, to satisfy even modest information
needs, the user is often required to formulate her queries using one
of existing query languages whose common core is
XPath~\cite{XPath1,XPath2}. XPath queries allow to access the contents
of the desired nodes with a syntax similar to directory paths used to
navigate in the UNIX file system. Unfortunately, even the XPath query
language, and any language with formal syntax, might be too difficult
to be accessible to every user, and in general, there is a lack of
frameworks allowing the user to formulate the query without the
knowledge of a specialized query language.

In this paper, we propose to address this gap with the help of
algorithms that infer the query from examples given by the user. We
remark, however, that the need for general inference of XML queries is
justified by other novel database applications. For instance, in the
setting of XML data exchange~\cite{ArLi05} the pattern queries used to
define data mappings need to be specified by the user. A learning
algorithm could be a base for real \emph{ad-hoc} data exchange
solutions, where the pattern queries defining mappings are inferred as
new sources are discovered. Another example of potential application
is wrapper induction~\cite{GoCe09,Soderland99}.

The problem of XML query learning is defined as follows: given an XML
document with nodes annotated by the user construct a query that
selects the nodes accordingly to the annotations. Clearly, this
problem has two parameters: the class of queries within which the
algorithm should produce its result and the type of annotations the
user may use. In the current work we focus on two well-known
subclasses of XPath: \emph{twig} and \emph{path
  queries}~\cite{ACLS02}. We identify two types of annotations:
\emph{required} nodes i.e., nodes that need to be selected by the
query, and \emph{forbidden} nodes i.e., those that the query must not
select. Because we do not require all nodes to be annotated, every
unannotated node is implicitly annotated as \emph{neutral}, which
means that the query may or may not select it. In terms of
computational learning theory~\cite{KeVa94}, a required node is called
a \emph{positive example} and a forbidden node is a \emph{negative
  example}. In this paper, we consider two settings: one, where the
user provides only positive examples, and a more general one, where
both positive and negative examples are present.
\begin{example}
\label{ex:intro-unary}
Take for instance the XML document in Figure~\ref{fig:library} with a
library listing. Some of its elements are annotated as required ($+$)
and some as forbidden ($-$).
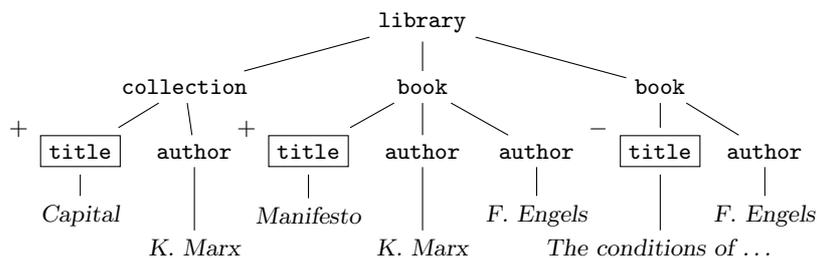
\begin{figure}[htb]
  \centering
  \begin{tikzpicture}[yscale=0.85,xscale=1.25]\small
    \node (r) at (0,0) {\tt library};
    \node (n1) at (-2.5,-1) {\tt collection} edge[-] (r);
    \node (n2) at (0,-1) {\tt book} edge[-] (r);
    \node (n3) at (2.5,-1) {\tt book} edge[-] (r);
    \node[draw,outer sep = 3pt] (n4) at (-3.6,-2) {\tt title} edge[-] (n1);
    \node at (-4.25,-1.65) {$+$};
    \node (n5) at (-2.4,-2) {\tt author} edge[-] (n1);
    \node[draw,outer sep = 3pt] (n6) at (-1.2,-2) {\tt title} edge[-] (n2);
    \node at (-1.85,-1.65) {$+$};
    \node (n7) at (-0,-2) {\tt author} edge[-] (n2);
    \node (n8) at (1.2,-2) {\tt author} edge[-] (n2);
    \node[draw,outer sep = 3pt] (n9) at (2.5,-2) {\tt title} edge[-] (n3);
    \node at (1.85,-1.65) {$-$};
    \node (n10) at (3.6,-2) {\tt author} edge[-] (n3);
    \node (n11) at (-3.6,-3) {\sl Capital} edge[-] (n4);
    \node (n12) at (-2.4,-3.5) {\sl K. Marx} edge[-] (n5);
    \node (n11) at (-1.2,-3) {\sl Manifesto} edge[-] (n6);
    \node (n12) at (-0,-3.5) {\sl K. Marx} edge[-] (n7);
    \node (n13) at (1.2,-3) {\sl F. Engels} edge[-] (n8);
    \node (n14) at (2.5,-3.5) {\sl The conditions of \ldots} edge[-] (n9);
    \node (n15) at (3.6,-3) {\sl F. Engels} edge[-] (n10);
  \end{tikzpicture}
  \caption{Annotation of a library database}
  \label{fig:library}
\end{figure}

The query that the user might want to receive is one that selects the
titles of works by K.\ Marx:
\[
q_0 = /\mathtt{library}/\wc[
\mathtt{author}\mathord{=}"\text{K.\ Marx}"]/
\mathtt{title}.
\]
The query
$/\mathtt{library}/\wc[\mathtt{author}\mathord{=}"\text{K.~Marx}"]/\wc$
is also consistent with the annotation but it properly contains $q_0$.
This makes $q_0$ more specific w.r.t.\ the user annotations, and
therefore, may be better fitted for the results of learning.  The
query selecting titles of all works,
$/\mathtt{library}/\wc/\mathtt{title}$ is not consistent because it
selects the forbidden \texttt{title} node. The query
$/\mathtt{library}/\mathtt{book}[\mathtt{author}\mathord{=}"\text{K.~Marx}"]/\mathtt{title}$
is also not consistent with the annotation because it does not select
the required \texttt{title} node of \textsl{Capital}. \qed
\end{example}

Our study requires us to define precisely what it means for a class of
queries $\mathcal{Q}$ to be \emph{learnable}. We propose a definition
influenced by computational learning theory~\cite{KeVa94}, and
inference of languages in particular~\cite{Go67,OnGa91,Hi97}. First of
all, for $\mathcal{Q}$ to be learnable there must exist a learning
algorithm $\learner$ which on the input takes a sample $S$ i.e., a set
of examples, and returns a query $q\in\mathcal{Q}$. Naturally,
$\learner$ should be \emph{sound}, that is the query $q$ must be
consistent with the sample $S$. Because the soundness condition is not
enough to filter out trivial learning algorithms (cf. discussion
following Definition~\ref{def:1}), we furthermore require $\learner$
to be \emph{complete}, that is able to learn every query with
sufficiently informative examples. More precisely, $\learner$ is
\emph{complete} if for every $q\in\mathcal{Q}$ there exists a so
called \emph{characteristic sample} $\CS_q$ of $q$ (w.r.t.\
$\learner$) such that $\learner(\CS_q)$ returns $q$. Note that an
unsavy user in the role of a teacher may not know exactly what is the
characteristic sample, but rather attempt to approach it by adding
more and more examples until the algorithm returns a satisfactory
query. Consequently, it is commonly required for the characteristic
sample to be \emph{robust} under inclusion i.e., $\learner(S)$ should
return $q$ for any sample $S$ that extends $\CS_q$ while being
consistent with $q$. Finally, polynomial restrictions are imposed on
$\learner$ and the size of the characteristic sample to ensure
tractability of the framework.

The primary goal of this paper is learning unary queries, but on the
way there we also investigate the learnability of Boolean
queries. Unary queries select a set of nodes in a document and are
typically used for information extraction tasks. On the other hand,
Boolean queries test whether or not a given document satisfies certain
property, and their typical use case is the classification of
documents e.g., for filtering purposes. When learning a Boolean query,
an example is a tree with a marker indicating whether it is a positive
or a negative example.
\begin{example}
\label{ex:intro-boolean}
Consider a simple XML feed with offers from a consumer-to-consumer web
site (Figure~\ref{fig:rss}) annotated by the user as either required
($+$) or forbidden ($-$).
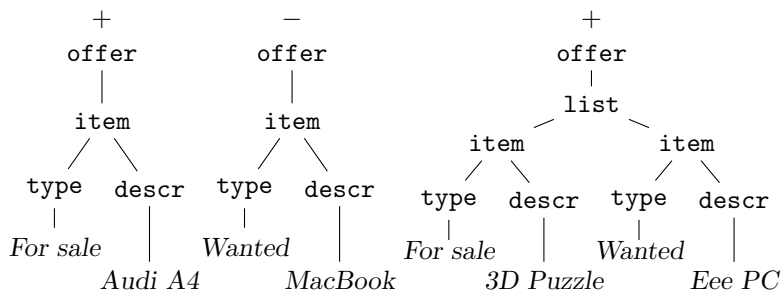
\begin{figure}[htb]
  \centering
  \begin{tikzpicture}[yscale=0.9,xscale=1.25]
    \begin{scope}[xshift=0cm]
      \node at (0,0.5) {$+$};
      \node (n0) at (0,0) {\tt offer};
      \node (n1) at (0,-1) {\tt item} edge[-] (n0);
      \node (n2) at (-0.5,-2) {\tt type} edge[-] (n1);
      \node (n4) at (0.5,-2) {\tt descr} edge[-] (n1);
      \node (n3) at (-0.5,-2.85) {\sl For sale} edge[-] (n2);
      \node (n5) at (0.5,-3.35) {\sl Audi A4} edge[-] (n4);
    \end{scope}
    \begin{scope}[xshift=2cm]
      \node at (0,.5) {$-$};
      \node (n0) at (0,0) {\tt offer};
      \node (n1) at (0,-1) {\tt item} edge[-] (n0);
      \node (n2) at (-0.5,-2) {\tt type} edge[-] (n1);
      \node (n4) at (0.5,-2) {\tt descr} edge[-] (n1);
      \node (n3) at (-0.5,-2.85) {\sl Wanted} edge[-] (n2);
      \node (n5) at (0.5,-3.35) {\sl MacBook} edge[-] (n4);
    \end{scope}
    \begin{scope}[xshift=5.15cm]
      \node at (0,.5) {$+$};
      \node (n0) at (0,0) {\tt offer};
      \node (n0') at (0,-0.75) {\tt list} edge[-] (n0);
      \node (n1) at (-1,-1.33) {\tt item} edge[-] (n0');
      \node (n2) at (-1.5,-2.17) {\tt type} edge[-] (n1);
      \node (n4) at (-0.5,-2.17) {\tt descr} edge[-] (n1);
      \node (n1') at (1,-1.33) {\tt item} edge[-] (n0');
      \node (n2') at (0.5,-2.17) {\tt type} edge[-] (n1');
      \node (n4') at (1.5,-2.17) {\tt descr} edge[-] (n1');
      \node[inner sep=0pt] (n3) at (-1.5,-2.9) {\sl For sale} edge[-] (n2);
      \node (n5) at (-0.5,-3.35) {\sl 3D Puzzle} edge[-] (n4);
      \node[inner sep=0pt] (n3') at (0.5,-2.9) {\sl Wanted} edge[-] (n2');
      \node (n5') at (1.5,-3.35) {\sl Eee PC} edge[-] (n4');
    \end{scope}
  \end{tikzpicture}
  \caption{An annotated XML stream.}
  \label{fig:rss}
\end{figure}
A Boolean query satisfying the user annotations selects all sale
offers i.e.,
$q_1=.[\mathtt{offer}\dblslash{}\mathtt{item}/\mathtt{type}\mathord{=}"\text{For
  sale}"]$.\qed
\end{example}

We investigate the learnability for Boolean and unary, path and twig
queries in the presence of positive examples only and in the presence
of both positive and negative examples. For learning in the presence
of positive examples only, we identify practical subclasses of
\emph{anchored} path queries and \emph{path-subsumption-free} twig
queries that are learnable. The main idea behind our learning
algorithms is to attempt to construct an (inclusion-)minimal query
consistent with the examples. Intuitively this means that our
algorithms try to construct a query that is as specific as possible
with respect to the user input (cf. $q_0$ in
Example~\ref{ex:intro-unary}). This approach is common to a host of
algorithms learning concepts from positive examples~\cite{Angluin80}
including reversible regular languages~\cite{Angluin82}, $k$-testable
regular languages~\cite{GaVi90}, and single occurrence regular
expressions~\cite{BeGeNeVa10}. While our learning algorithms for path
queries return minimal queries consistent with the input sample, we
show that this approach cannot be fully adopted for twig queries
because there are input samples for which the consistent minimal twig
query is of exponential size. Here, our learning algorithms return
queries that can be seen as polynomially-sized approximations.

The learnability of the full classes of path and twig queries remains
an open question. However, we identify the essential properties of the
query classes that enable our learning techniques, and observe that
these properties do not hold for the full classes of path and twig
queries. This indicates that new approaches may need to be explored if
learning of the full classes is feasible at all.
 
In the setting where both positive and negative examples are allowed,
we study the \emph{consistency problem}: given a document with a set
of positive and negative annotations is there a query that satisfies
the annotations? This problem is trivial if only positive examples are
given because the universal query, that selects all nodes in a tree,
is consistent with any set of positive examples. However, as we show,
adding even one negative example renders the consistency problem
intractable. This result holds for all considered classes of queries,
including anchored path queries and path-subsumption-free twig
queries, and in fact, it holds for so simple classes of queries that
it is hard to envision some reasonable restrictions that would admit
learnability in the presence of positive and negative examples.

The main contribution of this paper is defining and establishing
theoretical boundaries for learning path and twig queries from
examples. To the best of our knowledge this is the first work
addressing this particular problem. Additionally, we investigate two
problems that might be of independent interest: constructing a minimal
query consistent with a set of positive examples and checking the
consistency of a set of positive and negative examples. The
characterization of the properties of the learnable classes of queries
and the algorithm for learning unary path queries are based on
existing techniques, tree pattern homomorphisms~\cite{MiSu99,MiSu04}
and pattern learning~\cite{Angluin79,Shinohara82}, but we employ them
in new, nontrivial ways. The remaining results, including the
remaining learning algorithms and intractability of the consistency
problem, are new and nontrivial.

The paper is organized as follows. In Section~\ref{sec:basic-notions}
we introduce basic notions and define formally the learning
framework. In Section~\ref{sec:classes} we define the learnable
subclasses of queries and identify their essential properties that
enable our learning algorithms. In Sections~\ref{sec:simple-path-expr}
through \ref{sec:tree-queries} we present the corresponding learning
algorithms. In Section~\ref{sec:impact-negat-exampl} we discuss the
impact of negative examples on learning. We discuss the related work
in Section~\ref{sec:related-work}. Finally, we summarize our results
and outline further directions in
Section~\ref{sec:concl-future-work}. Because of space restriction we
present only sketches of the most important proofs; complete proofs
will be given in the full version of the paper (currently in
preparation for journal submission).

\noindent {\bf Acknowledgments.} We would like to thank our fellow
colleagues and anonymous reviewers for their helpful comments. We also
would like to thank Radu Ciucanu and Ioana Adam who implemented the
algorithms and shared their insights allowing to improve theoretical
properties of the algorithms. This research has been partially
supported by Ministry of Higher Education and Research, Nord-Pas de
Calais Regional Council and FEDER through the Contrat de Projets Etat
Region (CPER) 2007-2013, Codex project ANR-08-DEFIS-004, and Polish
Ministry of Science and Higher Education research project N N206
371339.

\section{Basic notions}
\label{sec:basic-notions}

Throughout this paper we assume an infinite set of node labels
$\Sigma$ which allows us to model documents with textual values. We
also assume that $\Sigma$ has a total order, that can be tested in
constant time, and has a minimal element that can be obtained in
constant time as well. We extend the order on $\Sigma$ to the standard
lexicographical order $\leq_\lex$ on words over $\Sigma$ and define a
well-founded canonical order on words: $w\leq_\can u$ iff $|w|<|u|$ or
$|w|=|u|$ and $w\leq_\lex u$.

\paragraph{Trees} We model XML documents with unranked labeled
trees. Formally, a {\em tree} $t$ is a tuple
$(N_t,\root_t,\lab_t,\child_t)$, where $N_t$ is a finite set of nodes,
$\root_t\in N_t$ is a distinguished root node,
$\lab_t:N_t\rightarrow\Sigma$ is a labeling function, and
$\child_t\subseteq N_t\times N_t$ is the parent-child relation. We
assume that the relation $\child_t$ is acyclic and require every
non-root node to have exactly one predecessor in this relation. By
$\Tree_0$ we denote the set of all trees.

The \emph{size} of a tree is the cardinality of its node set. The
\emph{depth} of a node is the length of the path from the root to the
node and the \emph{height} of the tree is the depth of its deepest
leaf. For a tree $t$ by $\Paths(t)$ we denote the set of paths from
the root node to the leaf nodes of $t$. We view a path both as a tree,
in particular it has nodes, and as a word. Often, we use unranked
terms over $\Sigma$ to represent trees. For instance, the term
$r(a(b),b(a(c)),c(b(a)))$ corresponds to the tree $t_0$ in
Figure~\ref{fig:tree}.
\begin{figure}[htb]
  \centering
  \subfigure[Tree $t_0$.]{\label{fig:tree}
    \centering
    \begin{tikzpicture}[yscale=0.75,xscale=1.25]
      \node at (0,0) (n0) {$r$};
      \node at (-0.75,-1) (n1) {$a$};
      \node at (0,-1) (n4) {$b$};
      \node at (0,-2) (n5) {$a$};
      \node at (0.75,-1) (n2) {$c$};
      \node at (0.75,-2) (n3) {$b$};
      \node at (0.75,-3) (n6) {$a$};
      \node at (-0.75,-2) (n7) {$b$};
      \node at (0,-3) (n8) {$c$};
      \draw[-,semithick] (n5) -- (n8);
      \draw[-,semithick] (n1) -- (n7);      
      \draw[-,semithick] (n0) -- (n1);
      \draw[-,semithick] (n0) -- (n1);
      \draw[-,semithick] (n0) -- (n2);
      \draw[-,semithick] (n2) -- (n3);
      \draw[-,semithick] (n3) -- (n6);
      \draw[-,semithick] (n0) -- (n4);
      \draw[-,semithick] (n4) -- (n5);
    \end{tikzpicture} 
  }
  \hspace{3cm}
  \subfigure[Decorated trees $t_1$ and $t_2$.]{
    \label{fig:decorated-trees}
    \centering
    \begin{tikzpicture}[yscale=0.75,xscale=1.25]
      \node at (0,0) (n0) {$r$};
      \node at (-0.75,-1) (n1) {$a$};
      \node at (0,-1) (n4) {$b$};
      \node[draw,outer sep = 3pt] at (0,-2) (n5) {$a$};
      \node at (0.75,-1) (n2) {$c$};
      \node at (0.75,-2) (n3) {$b$};
      \node at (0.75,-3) (n6) {$a$};
      \node at (-0.75,-2) (n7) {$b$};
      \node at (0,-3) (n8) {$c$};
      \draw[-,semithick] (n5) -- (n8);
      \draw[-,semithick] (n1) -- (n7);      
      \draw[-,semithick] (n0) -- (n1);
      \draw[-,semithick] (n0) -- (n1);
      \draw[-,semithick] (n0) -- (n2);
      \draw[-,semithick] (n2) -- (n3);
      \draw[-,semithick] (n3) -- (n6);
      \draw[-,semithick] (n0) -- (n4);
      \draw[-,semithick] (n4) -- (n5);
      \begin{scope}[xshift=2.75cm]
      \node at (0,0) (n0) {$r$};
      \node at (-0.75,-1) (n1) {$a$};
      \node at (0,-1) (n4) {$b$};
      \node at (0,-2) (n5) {$a$};
      \node at (0.75,-1) (n2) {$c$};
      \node at (0.75,-2) (n3) {$b$};
      \node[draw,outer sep = 3pt] at (0.75,-3) (n6) {$a$};
      \node at (-0.75,-2) (n7) {$b$};
      \node at (0,-3) (n8) {$c$};
      \draw[-,semithick] (n5) -- (n8);
      \draw[-,semithick] (n1) -- (n7);      
      \draw[-,semithick] (n0) -- (n1);
      \draw[-,semithick] (n0) -- (n1);
      \draw[-,semithick] (n0) -- (n2);
      \draw[-,semithick] (n2) -- (n3);
      \draw[-,semithick] (n3) -- (n6);
      \draw[-,semithick] (n0) -- (n4);
      \draw[-,semithick] (n4) -- (n5);

      \end{scope}
    \end{tikzpicture} 
  }
  \caption{Trees.}
  \label{fig:trees}
\end{figure}

To represent examples and answers to queries, we use trees with one
distinguished \emph{selected} node. Formally, a {\em decorated tree}
is a pair $(t,\sel_t)$, where $t$ is a tree and $\sel_t\in N_t$ is a
distinguished {\em selected node}. We denote the set of all decorated
trees by $\Tree_1$. Figure~\ref{fig:decorated-trees} contains two
decorated versions of $t_0$: the selected node is indicated with a
square box. In the sequel, we rarely make the distinction between
standard trees and decorated ones, and when it does not lead to
ambiguity, we refer to both structures as simply trees.

\noindent {\bf Queries.} We work with the class of twig queries, also
know as \emph{tree pattern queries}~\cite{ACLS02}. Twig queries are
essentially unordered trees whose nodes may be additionally labeled
with a distinguished wildcard symbol $\wc$ and that use two types of
edges, child and descendant, corresponding to the standard XPath
axes. To model unary queries we also add a distinguished selecting
node.
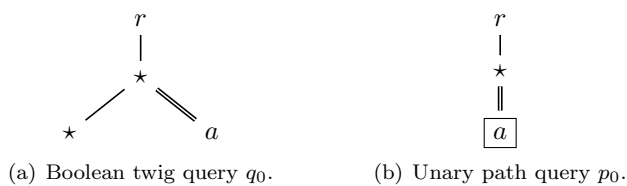
\begin{figure}[htb]
  \centering
  \subfigure[Boolean twig query $q_0$.]{\label{fig:twig0}
    \centering
    \begin{tikzpicture}[yscale=0.75, xscale=1.25]
      \path[use as bounding box] (-1.75,.25) rectangle (1.75,-2.25);
      \node at (0,0) (n0) {$r$};
      \node at (0,-1) (n2) {$\wc$};
      \node at (0.75,-2) (n3) {$a$};
      \node at (-0.75,-2) (n1) {$\wc$};
      \draw[-,semithick] (n2) -- (n1);
      \draw[-,semithick] (n0) -- (n2);
      \draw[-,double,semithick] (n2) -- (n3);
    \end{tikzpicture}
  }
  \subfigure[Unary path query $p_0$.]{\label{fig:twig1}
    \centering
    \begin{tikzpicture}[yscale=0.75,xscale=1.25]
      \path[use as bounding box] (-1.75,.25) rectangle (1.75,-2.25);
      \node at (0,0) (n0) {$r$};
      \node at (0,-0.9) (n2) {$\wc$};
      \node[draw,outer sep = 3pt] at (0,-2) (n3) {$a$};
      \draw[-,semithick] (n0) -- (n2);
      \draw[-,double,semithick] (n2) -- (n3);
    \end{tikzpicture}
  }
  \caption{Twig queries.}
  \label{fig:twigs}
\end{figure}

A \emph{Boolean twig query} $q$ is a tuple
$(N_q,\root_q,\lab_q,\child_q,\desc_q)$, where $N_q$ is a
finite set of nodes, $\root_q\in N_q$ is the root node,
$\lab_q:N_q\rightarrow\Sigma\cup\{\wc\}$ is a labeling function,
$\child_q\subseteq N_q\times N_q$ is a set of child edges, and
$\desc_q\subseteq N_q\times N_q$ is a set of descendant edges.  We
assume that $\child_q\cap\desc_q=\emptyset$ and that the relation
$\child_q\cup\desc_q$ is acyclic and require every non-root node to
have exactly one predecessor in this relation. By $\Twig_0$ we denote
the set of all Boolean twig queries. A \emph{unary twig query} is a
pair $(q,\sel_q)$, where $q$ is a Boolean twig query and $\sel_q\in
N_q$ is a distinguished \emph{selecting} node. We denote the set of
all twig queries by $\Twig_1$. Figure~\ref{fig:twigs} contains
examples of twig queries: child edges are drawn with a single line,
descendant edges with a double line, and the selecting node is
indicated with a square box.

Additionally, we use restricted classes of Boolean and unary path
queries, $\Path_0$ and $\Path_1$ respectively. Formally,
$\Path_i$ contains those elements of $\Twig_i$ whose nodes have at
most one child. Furthermore, the selecting node of a unary path query
is always its only leaf (cf.\ Figure~\ref{fig:twig1}). We note that
$\Twig_1$ captures exactly the class of descending positive
disjunction-free XPath queries, and in the sequel, we use elements of
the abbreviated XPath syntax~\cite{XPath1,XPath2} to present both
elements of $\Twig_1$ and $\Twig_0$. For instance, the query in
Figure~\ref{fig:twig0} can be written as $r/\wc[\wc]\dblslash{}a$, and
the query in Figure~\ref{fig:twig1} as $r/\wc\dblslash{}a$.

Because no unary twig query can select at the same time the root node
and another node of a tree, we disallow the root to be an answer, and
from now on, we consider only unary queries and decorated trees whose
selected node is other than root. Note that this restriction can be
easily bypassed by adding a virtual root node to every tree in the
input sample. Also, this way the \emph{universal query} is
$\wc\dblslash{}\wc$.

\paragraph{Embeddings} We define the semantics of twig queries
using the notion of embedding which is essentially a mapping of nodes
of a query to the nodes of a tree (or another query) that respects the
semantics of the edges of the query. In the sequel, for two
$x,y\in\Sigma\cup\{\wc\}$ we say that $x$ \emph{matches} $y$ if
$y\neq\wc$ implies $x=y$. Note that this relation is not symmetric:
$a$ matches $\wc$ but $\wc$ does not match $a$.

Formally, for $i\in\{0,1\}$, a query $q\in\Twig_i$ and a tree
$t\in\Tree_i$, an {\em embedding} of $q$ in $t$ is a function $\lambda
: N_q \rightarrow N_t$ such that:
\begin{enumerate}
\itemsep0pt
\item[$1$.] $\lambda(\root_q)=\root_t$,
\item[$2$.] for every $(n,n')\in\child_q$,
  $(\lambda(n),\lambda(n'))\in\child_t$,
\item[$3$.] for every $(n,n')\in\desc_q$,
  $(\lambda(n),\lambda(n'))\in(\child_t)^+$,
\item[$4$.] for every $n\in N_q$, $\lab_t(\lambda(n))$ matches
  $\lab_q(n)$,
\item[$5$.] if $i=1$, then $\lambda(\sel_q)=\sel_t$. 
\end{enumerate}
Then, we write $\lambda: q\hookrightarrow t$ or simply
$t\preccurlyeq{}q$. Figure~\ref{fig:embeddings} presents all
embeddings of the query $q_0$ in the tree $t_0$
(Figure~\ref{fig:tree}). \begin{figure}[htb] \centering
  \begin{tikzpicture}[yscale=0.75,xscale=1.25]
    \node at (0,0) (n0) {$r$};
    \node at (-0.75,-1) (n1) {$a$};
    \node at (0,-1) (n4) {$b$};
    \node at (0,-2) (n5) {$a$};
    \node at (0.75,-1) (n2) {$c$};
    \node at (0.75,-2) (n3) {$b$};
    \node at (0.75,-3) (n6) {$a$};
    \node at (-0.75,-2) (n7) {$b$};
    \draw[-,semithick] (n1) -- (n7);      
    \draw[-,semithick] (n0) -- (n1);
    \draw[-,semithick] (n0) -- (n2);
    \draw[-,semithick] (n2) -- (n3);
    \draw[-,semithick] (n3) -- (n6);
    \draw[-,semithick] (n0) -- (n4);
    \draw[-,semithick] (n4) -- (n5);
    \begin{scope}[xshift=-2.75cm]
      \node at (0,0) (m0) {$r$};
      \node at (0,-1) (m2) {$\wc$};
      \node at (0.75,-2) (m3) {$a$};
      \node at (-0.75,-2) (m1) {$\wc$};
      \draw[-,semithick] (m2) -- (m1);
      \draw[-,semithick] (m0) -- (m2);
      \draw[-,double,semithick] (m2) -- (m3);

    \draw (m0) edge[->,bend left,densely dotted] (n0);
    \draw (m1) edge[->,bend right,densely dotted] (n5);
    \draw (m2) edge[->,bend left,densely dotted] (n4);
    \draw (m3) edge[->,bend left,densely dotted] (n5);
    \end{scope}
    \begin{scope}[xshift=2.75cm]
      \node at (0,0) (m0) {$r$};
      \node at (0,-1) (m2) {$\wc$};
      \node at (0.75,-2) (m3) {$a$};
      \node at (-0.75,-2) (m1) {$\wc$};
      \draw[-,semithick] (m2) -- (m1);
      \draw[-,semithick] (m0) -- (m2);
      \draw[-,double,semithick] (m2) -- (m3);
      
    \draw (m0) edge[->,bend right,densely dotted] (n0);
    \draw (m1) edge[->,bend right,densely dotted] (n3);
    \draw (m2) edge[->,bend right,densely dotted] (n2);
    \draw (m3) edge[->,bend left,densely dotted] (n6);
    \end{scope}
  \end{tikzpicture}
  \caption{Embeddings of $q_0$ in $t_0$.}
  \label{fig:embeddings}
\end{figure}

Note that we do not require the embedding to be injective i.e., two
nodes of the query may be mapped to the same node of the
tree. Embeddings of path queries are, however, always injective. Also,
note that the semantics of $\dblslash$-edge is that of a proper
descendant (and not that of descendant-or-self).

Typically, the semantics of a unary query is defined in terms of the
set of nodes it selects in a tree~\cite{Li06b,MiSu04}: a node $n$ of a
tree $t$ is an \emph{answer} to a unary twig query $q$ in $t$ if there
is an embedding $\lambda:q\hookrightarrow t$ such that
$\lambda(\sel_q)=n$ (then $n$ is also said to be \emph{reachable} by
$q$ in $t$). However, we use an alternative way of defining the
semantics of a query. Formally, the \emph{language} of a query
$q\in\Twig_i$ for $i\in\{0,1\}$ is the set
\[
\mathcal{L}_i(q)=\{t\in\Tree_i\mid t\preccurlyeq q\}.
\]
Naturally, the two notions are very closely related e.g., the
decorated trees $t_1$ and $t_2$ (Figure~\ref{fig:trees}) belong to
$\mathcal{L}_1(p_0)$ (Figure~\ref{fig:twigs}) and the nodes selected in
$t_1$ and $t_2$ are exactly the answers to $p_0$ in tree $t_0$.

The notion of an embedding extends in a natural fashion to a pair of
queries $q,p\in\Twig_i$ for some $i\in\{0,1\}$: an \emph{embedding} of
$q$ in $p$ is a function $\lambda:N_q\rightarrow N_p$ that satisfies
the conditions $1$, $2$, $4$, $5$ above (with $t$ being replaced by
$p$) and the following condition:
\begin{enumerate}
\itemsep0pt
\item[$3\rlap{${}'$}$.] for all $(n,n')\in\desc_q$,
  $(\lambda(n),\lambda(n'))\in(\child_{p}\cup\desc_{p})^+$.
\end{enumerate}
Then, we write $\lambda:p\hookrightarrow q$ or simply $q\preccurlyeq
p$ and say that $p$ \emph{subsumes} $q$.

The \emph{containment} (or \emph{inclusion}) $q\subseteq p$ of two
queries $q,p\in Twig_i$ for $i\in\{0,1\}$ is simply
$\mathcal{L}_i(q)\subseteq\mathcal{L}_i(p)$, and we say that $q$ and
$p$ are \emph{equivalent}, denoted $q\equiv p$, if $q\subseteq p$ and
$p\subseteq q$. Note that for twigs, subsumption implies containment
i.e., if $q\preccurlyeq p$, then $q\subseteq p$. The converse does not
hold in general. For instance, we have $a[.\dblslash{}b]\subseteq\wc[\wc]$ but
$a[.\dblslash{}b]\not\preccurlyeq\wc[\wc]$. There are also significant
computational differences: the containment of twigs is
coNP-complete~\cite{Sc04,NeSc03} whereas their subsumption is in
PTIME.

\paragraph{Query minimality} In this paper we identify queries that are
minimal for a given set of trees (as examples). It is important to
emphasise that we always mean minimality in terms of query
inclusion. Formally, for $i\in\{0,1\}$, a class of queries
$\mathcal{Q}\subseteq\Twig_i$, a query $q\in\mathcal{Q}$, and a set of
trees $S\subseteq\Tree_i$, we say that $q$ is \emph{minimal query in
  $\mathcal{Q}$ consistent with $S$} if $S\subseteq\mathcal{L}_i(q)$
and there is no $q'\in\mathcal{Q}$ such that $q'\subseteq q$,
$q'\not\equiv q$, and $S\subseteq\mathcal{L}_i(q')$.

\paragraph{Learning framework} We use a variant of the standard
language inference framework~\cite{KeVa94,Go67,OnGa91,Hi97} adapted to
learning queries. A learning setting comprises of the set of concepts
that are to be learnt, in our case queries, and the set of instances
of the concepts that are to serve as examples in learning, in our case
trees (possibly decorated). These two sets are bound together by the
semantics which maps every concept to its set of instances.
\begin{definition}
  A {\em learning setting} is a tuple
  $(\mathcal{D},\mathcal{Q},\mathcal{L})$, where $\mathcal{D}$ is a
  set of examples, $\mathcal{Q}$ is a class of queries, and
  $\mathcal{L}$ is a function that maps every query in $\mathcal{Q}$
  to the set of all its examples (a subset of $\mathcal{D}$). \qed
\end{definition}
As an example, a setting for learning unary Twig queries from positive
examples is the tuple $(\Tree_1,\Twig_1,\mathcal{L}_1)$. This general
formulation allows also to easily define settings for learning from
both positive and negative examples, which we present in
Section~\ref{sec:impact-negat-exampl}.

To define formally what learnability for queries means we fix a
learning setting $\mathcal{K}=(\mathcal{D},\mathcal{Q},\mathcal{L})$
and introduce some auxiliary notions. A \emph{sample} is a finite
nonempty subset $S$ of $\mathcal{D}$ i.e., a set of examples. The
\emph{size} of a sample is the sum of the sizes of the examples it
contains. A sample $S$ is \emph{consistent} with a query
$q\in\mathcal{Q}$ if $S\subseteq\mathcal{L}(q)$. A \emph{learning
  algorithm} is an algorithm that takes a sample and returns a query
in $\mathcal{Q}$ or a special value $\Null$.
\begin{definition}
  \label{def:1}
  A query class $\mathcal{Q}$ is \emph{learnable in polynomial time
    and data} in the setting
  $\mathcal{K}=(\mathcal{D},\mathcal{Q},\mathcal{L})$ iff there exits a
  polynomial learning algorithm $\learner$ and a polynomial $\poly$
  such that the following two conditions are satisfied:
  \begin{enumerate}
    \itemsep0pt
  \item \textbf{Soundness.} For any sample $S$ the algorithm
    $\learner(S)$ returns a query consistent with $S$ or a special
    $\Null$ value if no such query exists.
  \item \textbf{Completeness.} For any query $q\in\mathcal{Q}$ there
    exists a sample $\CS_q$ such that for every sample $S$ that
    extends $\CS_q$ consistently with $q$ i.e., $\CS_q\subseteq
    S\subseteq\mathcal{L}(q)$, the algorithm $\learner(S)$ returns a
    query equivalent to $q$. Furthermore, the size of $\CS_q$ is
    bounded by $\poly(|q|)$.\qed
  \end{enumerate}
\end{definition}
The sample $\CS_q$ is often called the \emph{characteristic sample}
for $q$ w.r.t.\ $\learner$ and $\mathcal{K}$ but we point out that for
a learning algorithm there may exist many samples fitting the role and
the definition of learnability requires merely that one such sample
exists. The soundness condition is a natural requirement but alone it
is insufficient to eliminate trivial learning algorithms. For
instance, for the setting where only positive examples are used, an
algorithm returning the universal query $\wc\dblslash{}\wc$ is
sound. Consequently, we require the algorithm to be complete
analogously to how it is done for grammatical language
inference~\cite{Go67,OnGa91,Hi97}. An alternative and natural way to
ban trivial learning algorithms would be to require the algorithm to
return some minimal query consistent with the input sample. Our
approach follows this direction but as we show later on, it is not
possible to fully adhere to it because there exist samples for which
the minimal consistent twig query is of exponential size.

\section{Learnable query classes}
\label{sec:classes}
In this section we define the classes of queries, that in the
following sections we prove learnable from positive examples, and
identify two essential properties of these classes that enable our
learning algorithms. Both properties follow from the importance of
logical implication in learning: learning can often be seen as a
search of the correct hypothesis obtained by an iterative refinement
of some initial hypothesis and at every iteration the current
hypothesis is often a logical consequence of the previous one. The
first property requires the containment to be equivalent to
subsumption, which allows to capture containment with a simple
structural characterization. The second property is the existence of
polynomially sized \emph{match sets}~\cite{MiSu04}, which were
originally introduced as an easy way of testing query inclusion. The
match sets that we construct will serve us as the characteristic
samples. We emphasise that the full classes of twig and path queries
do not have these properties but this does not imply that they are not
learnable but it merely precludes the direct adaptation of our
learning techniques. Whether the full classes of queries are learnable
remains an open question.

To formally define the two properties, we fix a class of queries
$\mathcal{Q}$ with their semantics defined by $\mathcal{L}$. The
properties are:
\begin{itemize}
\item[$(\Prop1)$] for every two $q_1,q_2\in\mathcal{Q}$, $q_1
  \subseteq q_2$ if and only if $q_1\preccurlyeq q_2$.
\item[$(\Prop2)$] every $q\in\mathcal{Q}$ has a polynomial
  \emph{match set} i.e., a set $\CS_q$ of (positive) examples such
  that the size of $\CS_q$ is polynomial in the size of $q$ and for
  every $q'\in \mathcal{Q}$ we have $q\subseteq q'$ if and only if
  $\CS_q\subseteq\mathcal{L}(q')$.
\end{itemize}
We next present the construction of match sets in a generic form and
then we introduce the learnable classes of queries and state the
properties $\Prop1$ and $\Prop2$ for them.

\subsection{Match sets as characteristic samples}
\label{sec:match-sets-char}
We now present the construction of match sets that will be later on
used as characteristic samples. Because the constructions of the match
sets for all the subclasses of queries are very similar, we present it
in a generic form. Take a twig query $q$, let $N$ be the size of $q$,
$a_0$ be the minimal element of $\Sigma$, and $a_1$ and $a_2$ be two
fresh symbols not used in $q$ and different from $a_0$. The
constructed match set $\CS_q$ contains exactly two trees: $t_0$ is
obtained from $q$ by replacing every $\wc$ with $a_0$ and every
descendant edge by a child edge; $t_1$ is obtained from $q$ by
replacing every $\wc$ with $a_1$ and every descendant edge with a path
of length $N$ whose all nodes are labeled with
$a_2$. Figure~\ref{fig:characteristic-sample} contains the
characteristic sample for the unary twig query
$q_1=r/b[a\dblslash{}b]\dblslash{}c[d]/\wc/c$.
\begin{figure}[htb]
  \centering
  \begin{tikzpicture}[yscale=0.65,xscale=1.25]
    \node at (-0.5,0) (q) {$q_1:$};
    \node at (0,0) (r) {$r$};
    \node at (0,-1) (n1) {$b$}     edge[-] (r);
    \node at (0.5,-2) (m1) {$a$} edge[-] (n1);
    \node at (0.5,-3) (m2) {$b$} edge[-,double] (m1);
    \node at (0,-2) (n2) {$c$}     edge[-,double] (n1);
    \node at (-0.5,-3) (n3') {$d$} edge[-] (n2);
    \node at (0,-3) (n3) {$\wc$}   edge[-] (n2);
    \node[draw,outer sep = 3pt] at (0,-4) (n4) {$c$}   edge[-] (n3);
    \begin{scope}[xshift=2cm]
      \node at (-0.5,0) (q) {$t_0:$};
      \node at (0,0) (r) {$r$};
      \node at (0,-1) (n1) {$b$}     edge[-] (r);
      \node at (0.5,-2) (m1) {$a$} edge[-] (n1);
      \node at (0.5,-3) (m2) {$b$} edge[-] (m1);
      \node at (0,-2) (n2) {$c$}     edge[-] (n1);
      \node at (-0.5,-3) (n3') {$d$}     edge[-] (n2);
      \node at (0,-3) (n3) {$a_0$}   edge[-] (n2);
      \node[draw,outer sep = 3pt] at (0,-4) (n4) {$c$}   edge[-] (n3);
    \end{scope}
    \begin{scope}[xshift=4cm]
      \node at (-0.5,0) (q) {$t_1:$};
      \node at (0,0) (r) {$r$};
      \node at (0,-1) (n1) {$b$}     edge[-] (r);
      \node at (0.5,-2) (m1) {$a$} edge[-] (n1);
      \node at (0.5,-3) (m') {$a_2$} edge[-] (m1);
      \node at (0.5,-4) (m'') {$a_2$} edge[-,dotted] (m');
      \node at (0.5,-5) (m2) {$b$} edge[-] (m'');
      \node at (0,-2) (n') {$a_2$} edge[-] (n1);
      \node at (0,-3) (n1) {$a_2$} edge[-,dotted] (n');
      \node at (0,-4) (n2) {$c$}     edge[-] (n1);
      \node at (-0.5,-5) (n3') {$d$}     edge[-] (n2);
      \node at (0,-5) (n3) {$a_1$}   edge[-] (n2);
      \node[draw,outer sep = 3pt] at (0,-6) (n4) {$c$}   edge[-] (n3);
      \node at (-0.5,-2.5) (brace) {\small$8\Big\{$};
      \node at (1,-3.5) (brace) {\small$\Big\}8$};
    \end{scope}
  \end{tikzpicture}
  \caption{The characteristic sample for $q_1$.}
  \label{fig:characteristic-sample}
\end{figure}
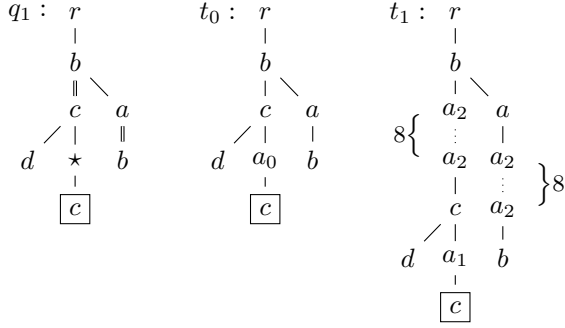
We point out that for a query and a learning algorithm there might be
more than just one characteristic sample. This is also the case with
our learning algorithms. While the construction we present above might
seem quite artificial, we use it due to its properties that might be
of independent interest (match sets). Simpler, and easier to compose
by a unskilled user, characteristic samples are often possible.

\subsection{Anchored path queries}
We begin with a base subclass of path queries, called anchored path
queries. Essentially, a path query is anchored when no inner $\wc$
node is incident to a $\dblslash{}$-edge.  The main reason for
introducing this class of queries is that when working with their
embeddings the restriction on the use of $\dblslash{}$ allows us to
limit the ``jumps'' that the embedding may perform in between two
nodes connected by a descendant edge. An additional restriction on the
leaf node of Boolean path queries is imposed for technical reasons
(cf. proof of Lemma~\ref{lemma:prop-anchored-paths} for more details).

Formally, the class of \emph{unary anchored path queries} imposes one
restriction: a $\dblslash{}$-edge cannot be incident to a $\wc$-node
unless it is the root node or the leaf node (which is also
selecting). For instance, the unary queries
$r\dblslash{}a\dblslash{}b/\wc/c$, $\wc\dblslash{}a\dblslash{}b/\wc$,
and $\wc\dblslash{}\wc$ are anchored but the query
$r\dblslash{}a/\wc\dblslash{}b$ is not. An additional restriction is
imposed on the \emph{Boolean anchored path queries}: if the leaf node
is $\wc$, then the edge incident to it is $\dblslash{}$. For instance,
the Boolean queries $a\dblslash{}b/\wc/c\dblslash{}\wc$ and
$a\dblslash{}b/\wc/c\dblslash{}a/\wc/b$ are anchored but the Boolean
query $\wc\dblslash{}a\dblslash{}b/\wc$ is not anchored. We denote by
$\APath_1$ and $\APath_0$ the sets of unary and Boolean anchored path
queries respectively.

Clearly, the subclasses of anchored path queries are properly included
in the full classes of path queries, however, we believe that the
restrictions are not very limiting and the classes of anchored queries
remain practical. Basically, anchored path queries cannot discriminate
the descendants of a node based on their depth alone. We also point
out that the additional restriction imposed on Boolean queries is
quite minor: the Boolean query $r\dblslash{}a/\wc$ is not anchored but
it is equivalent to $r\dblslash{}a\dblslash{}\wc$ which is
anchored. Note, however, that the Boolean query
$r\dblslash{}a/\wc/\wc$ does not have an equivalent Boolean anchored
query.

While \Prop1 for anchored path queries follows from the results
in~\cite{MiSu99,MiSu04}, below we present a proof using a technique
that allows to show \Prop1 and \Prop2 for all the query classes we
introduce later on (and these results are new and cannot be derived
from the results in~\cite{MiSu99,MiSu04}).
\begin{lemma}
  \label{lemma:prop-anchored-paths}
  Unary and Boolean anchored path queries have the properties \Prop1
  and \Prop2.
\end{lemma}
To prove this lemma it suffices to show the following claim.
\begin{claim}
  \label{claim:1}
  For any $i\in\{0,1\}$ and any two $q,q'\in\APath_i$, if
  $\CS_q\subseteq \mathcal{L}_i(q')$, then there exists an embedding
  $\lambda: q'\hookrightarrow q$.
  \label{claim:propAnchPaths}
\end{claim}
\begin{proof}
  We first give an equivalent yet more structured definition of
  anchored path queries. A \emph{block} is a path query fragment $B$
  of the form $\sigma_0/\ldots/\sigma_n$, where $n\geq 0$,
  $\sigma_0,\sigma_n\in\Sigma$, and
  $\sigma_1,\ldots,\sigma_{n-1}\in\Sigma\cup\{\wc\}$.  An
  \emph{anchored path query} $q$ is a path query of the form
  $B_0\dblslash{}B_1\dblslash{}\ldots\dblslash{}B_k$, where $k\geq 0$,
  $B_i$ is a block for $1\leq i \leq k-1$, and $B_0$ is either a block
  that can start with $\wc$ or a single occurrence of $\wc$. Also, in
  case of Boolean anchored path queries $B_k$ is either a block or a
  single occurrence of $\wc$ and in case of unary anchored path
  queries $B_k$ is either a block that can end with $\wc$ or a single
  occurrence of $\wc$.

  We first prove the claim for unary queries (i.e., $i=1$). Let
  $N=|q|$ and $\CS_q=\{t_0,t_1\}$ be constructed as described in
  Section \ref{sec:match-sets-char}. For every node $n$ of $t_1$ whose
  label is not $a_2$ by $\origin(n)$ we denote the node of $q$
  corresponding to $n$. Also, fix $\lambda_1:q'\hookrightarrow t_1$.

  We make several observations. First, $|q'|\leq N$, or otherwise
  there would be no embedding of $q'$ into $t_0$. For the same reason,
  $q'$ does not use the labels $a_1$ and $a_2$. Therefore, if a node
  $n$ of $q'$ is mapped by $\lambda_1$ to a node with label $a_1$ or
  $a_2$, then $\lab_{q'}(n)=\wc$. 

  Next, we show that $\lambda_1$ maps nodes of $q'$ only to those
  nodes of $t_1$ that are not labeled with $a_2$. This is clearly the
  case for the root node and the selecting node of $q'$, that are
  mapped to the root node and the selecting node of $t_1$, and from
  the construction of $t_1$, they have labels different from $a_2$. In
  the following we show the this is the case with other nodes. 

  Let $q'$ be of the form
  $B_0\dblslash{}B_1\dblslash{}\ldots\dblslash{}B_k$. Note that if a
  node $n$ is on the border of $B_j$ (for $0\leq j\leq k$) then from
  the definition of a block $n$ cannot be mapped to $a_2$.  This is
  because $n$ is either a root node, or a selecting node or its label
  is not $\wc$.

  Suppose, that some node of $q'$ belonging to $B_i$ for $0\leq i\leq
  k$, is mapped to a node with label $a_2$ and let $n_1$ and $n_2$ be
  the nodes that are on the borders of $B$. Because $|B|\leq|q'|\leq
  N$ and in $t_1$ nodes labeled with $a_2$ come in sequences of length
  $N$, one of the nodes $n_1$ and $n_2$ needs to be mapped to a node
  labeled with $a_2$. This implies that one of $n_1$ and $n_2$ is
  labeled with $\wc$; a contradiction.

  This shows that $\lambda=\lambda_1\circ\origin$ is a properly
  defined function mapping $N_{q'}$ to $N_q$. We now show that $\lambda$
  is an embedding of $q'$ into $q$. The condition 2 holds because
  $\lambda_1$ preserves the child relation and if nodes $(n_1, n_2)$
  are in $\child_{t_1}$ and both are not labelled with $a_2$ then
  $(\origin(n_1), \origin(n_2))\in \child_{q}$.  The conditions 1, 3,
  and 5 follow from the definition of $\lambda$. For the condition 4,
  take any $n\in N_{q'}$ such that $\sigma=\lab_{q'}(n)\neq\wc$ and
  note that then, $\lambda_1(n)$ in $t_1$ has the same label $\sigma$
  which is different from $a_1$ (because $q$ does not use $a_1$) and
  $a_2$ (as shown above). Therefore the node of $q$ that corresponds
  to $\lambda_1(n)$ is labeled with $\sigma$ as well.

  The proof for Boolean anchored path queries is analogous and it
  suffices to consider the case when $B_k$ is a single occurrence of
  $\wc$. Then indeed an embedding $\lambda_1: q'\hookrightarrow t_1$
  may map the $\wc$-leaf to a node labeled with $a_2$. We note,
  however, that $\lambda_1$ can be easily altered to map the
  $\wc$-leaf to a non $a_2$-node because the $\wc$-leaf is connected
  to with descendant edge and every $a_2$ node in $t_1$ has a
  descendant that is not labeled with $a_2$.\qed
\end{proof}
To show $\Prop1$ it is enough to show the implication from left to
right. Assume $q \subseteq q'$ and note that $\CS_{q} \subseteq
\mathcal{L}(q)$. Therefore, $\CS_{q} \subseteq \mathcal{L}(q')$, which
by Claim~\ref{claim:1}, gives us $q \preccurlyeq q'$. $\Prop2$ follows
directly from Claim~\ref{claim:1}.

\subsection{Conjunctions of anchored path queries}
In our approach to learn twig queries we use path learning algorithms
to infer a set of path queries satisfied in the input sample and then
we combine these path queries into a twig query. Therefore, the
midpoint between learning path queries and twig queries is learning
conjunctions of path queries. We apply this technique only to learn
Boolean twig queries and so we focus only on learning Boolean
conjunctions of path queries. For convenience we use sets of Boolean
path queries to represent conjunctions but a conjunction can also be
seen as a Boolean twig query consisting of path queries meeting at the
root node. The second representation is used to define the semantics
of conjunctions and their characteristic samples.

Because our path learning algorithms infer anchored queries, we
consider only conjunctions of Boolean anchored path queries. Also, if
we have inferred two Boolean path queries $p_1$ and $p_2$, and $p_1$
subsumes $p_2$, then from the point of learning there is no point in
keeping $p_2$ because $p_1$ contains more specific information and
makes $p_2$ redundant. Consequently, we consider only \emph{reduced}
conjunctions i.e., having no two different $p_1,p_2$ such that
$p_1\subseteq p_2$. Naturally the conjunctions must be also
\emph{head-consistent} i.e., any two paths queries in a conjunction
much have the same root label or otherwise we would not be able to
represent it as a twig query. By $\CPath_0$ we denote the class of
conjunctions of Boolean anchored path queries satisfying the
restrictions described above. The use of anchored path queries allows
to prove the following lemma in a manner analogous to the proof of
Lemma~\ref{lemma:prop-anchored-paths}.
\begin{lemma}
  \label{lemma:prop-conj-path}
  Conjunctions of Boolean anchored path queries have the properties
  \Prop1 and \Prop2.
\end{lemma}

\subsection{Path-subsumption-free twig queries}
As mentioned previously, our learning algorithms for twig queries
attempt to construct the query $q$ by combining the path queries from
a conjunction inferred beforehand. Because we infer a reduced set of
path queries, the constructed Boolean twig query $q$ has no two
$p_1,p_2\in\Paths(q)$ such that $p_1\subseteq p_2$, where $\Paths(q)$
is the set of Boolean path queries on paths from the root to all
leaves of $q$. Naturally, all path queries in $\Paths(q)$ need to be
anchored. Formally, a Boolean twig query $q$ is
\emph{path-subsumption-free} iff $\Paths(q)$ is a reduced set of
Boolean anchored path queries and by $\PTwig_0$ we denote the class of
Boolean path-subsumption-free twig queries.

The restrictions are relaxed slightly for unary twig queries and
reflect our learning algorithm that first infers a unary anchored
path, and next, decorates it with elements of $\PTwig_0$ used as
filter expressions. Recall that the selecting path in a unary twig
query is the path query on the path from the root node to the
selecting node. Formally, a unary twig query $q$ is
\emph{path-subsumption-free} iff the unary path query from the root
node to the selecting node of $q$ is anchored and every Boolean path
query on the path ending at a (non-selecting) leaf node and beginning
at the closest node on the selecting path is anchored. By $\PTwig_1$
we denote the class of unary path-subsumption-free twig queries.

The classes of path-subsumption-free twig queries may seem at first
very limited. We note, however, that a twig query belongs to our class
if every leaf label is different or every pair of leaves with the same
label cannot be compared with $\preccurlyeq$ (and all paths are
anchored). This simple sufficient condition yields a rather large
class of twig queries used in practice, especially if we consider the
following remark. One of the advantages of considering an infinite set
of labels $\Sigma$ is the ability to capture textual values (stored in
the leaves of a tree). Then, non-selecting leaves of tree patterns are
used for equality tests of text values, and rarely the same value is
used to make an equality test (on similar paths).

\begin{lemma}
  \label{lemma:prop-psf-twig}
  Path-subsumption-free twig queries have the properties \Prop1 and
  \Prop2.
\end{lemma}
We point out that in the proof of Lemma~\ref{lemma:prop-psf-twig} we
only use the fact that path-subsumption-free twig queries are
constructed from anchored paths. The other restriction, namely that
the path queries in $\Paths(q)$ cannot subsume one another, is not
used in this proof and it is essential only for the proper work of our
learning algorithms. Our recent results show that learning is possible
without this restriction and we intend to present these findings in
the journal version of the paper.

\section{Learning unary path queries}
\label{sec:simple-path-expr}

In Figure~\ref{fig:inference-apath1} we present a learning algorithm
for the learning setting $\APATH_1=(\Tree_1,\APath_1,\mathcal{L}_1)$
inspired by and extending several learning algorithms for regular
string patterns~\cite{Angluin79,Shinohara82}
(cf. Section~\ref{sec:related-work} for more details). Recall that
$\SelPath(t)$ stands for the path from the root node to the selecting
node of $t$ and extend it to samples $\SelPath(S)=\{\SelPath(t)\mid
t\in S\}$. The algorithm begins with a universal path query $\wc
\dblslash{} \wc$ and considers only the paths from the root to the
selected nodes in the input sample. It constructs the path query in
three stages.

\begin{figure}[htb]
  \resetLineNoCounter
  \begin{tabbing}
    xxx\=xx\=xx\=xx\=xx\=xx\=xx\=xx\=mmm\kill
    {\bf algorithm} $\learner\APATH_1(S)$\\
    {\bf Input:} a sample $S\subseteq\Tree_1$ of decorated trees\\
    {\bf Output:} a minimal $p\in\APath_1$ such that $S \subseteq \mathcal{L}_1(p)$\\
    \lineNo \> $w \colonequals \min_{\leq_{can}}(\SelPath(S))$\\ 
    \lineNo \> {\bf let} $w$ be of the form $a_0/a_1/\cdots/a_n$\\
    \lineNo \> $p \colonequals \wc \dblslash{} \wc $\\
    \lineNo \>{\bf foreach} subpath $u$ of $a_1/a_2/\cdots/a_{n-1}$\\
    \>in the order of decreasing lengths { \bf do }\\
    \lineNo \> \> replace in $p$ any $\dblslash{}$-edge by 
    $\dblslash{}u\dblslash{}$ 
    as long as $S\subseteq\mathcal{L}_1(p)$\\
    \lineNo \> {\bf let} $p$ be of the form $b_0\dblslash{}p_0\dblslash{}b_1$\\
    
    \lineNo \> {\bf if} $S\subseteq\mathcal{L}_1(p\{b_0\gets a_0\})$ {\bf then}\\
    \lineNo \> \> $p\colonequals p\{b_0\gets a_0\}$\\
    \lineNo \> {\bf if} $S\subseteq\mathcal{L}_1(p\{b_1\gets a_n\})$ {\bf then}\\
    \lineNo \> \> $p\colonequals p\{b_1\gets a_n\}$\\

    \lineNo \> {\bf foreach} descendant edge $\alpha$ in $p$ {\bf do}\\

    \lineNo \> \> find maximal $\ell$ s.t.\  $S\subseteq \mathcal{L}_1(p\{\alpha\gets\dblslash{}(\wc/)^{\ell}\})$\\
    \lineNo \> \> {\bf if} $S\subseteq \mathcal{L}_1(p\{\alpha\gets/(\wc/)^{\ell}\})$ {\bf then} \\
    \lineNo \> \> \> $p \colonequals p\{\alpha\gets/(\wc/)^{\ell}\}$\\
    \lineNo \> {\bf return} $p$
  \end{tabbing}
  \caption{Learning algorithm for $\APATH_1$.}
  \label{fig:inference-apath1}
\end{figure}

In the first stage (lines 4 through 6) the algorithm attempts to
identify a collection of factors, essentially path fragments, that are
mutually common to every path in $\SelPath(S)$. Note that if a factor
is present in every path, then it is also present in the
$\leq_\can$-minimal path $w$. The candidate query $p$ is gradually
refined with the factors and the invariant is that these factors are
mutually present on every path in $\SelPath(S)$ and in the specified
order. For every new candidate $w'$, $\learner\APATH_1$ attempts to
find a place where $w'$ can be inserted and yield a path query $p'$
consistent with $S$.

In the second stage (lines 8 through 11), the algorithm takes the
query $p$ and attempts to specialise the first and the last
occurrences of wildcard i.e., replace them with the corresponding
symbol taken from $w$. Here, $p\{x\gets e\}$ creates a copy of $p$ and
replaces in it the reference $x$ by expression $e$ (the original $p$
remains unchanged). In the third stage (lines 12 through 16) the
algorithm attempts to specialize every $\dblslash{}$-edge in $p$ i.e.,
replace it with a maximally long sequence $/\wc/\wc/\ldots/\wc$.

\begin{example}
\label{ex:execution-sample}
In this example we show the execution of the algorithm
$\learner\APATH_1$ on the sample $\{t_1,t_2,t_3\}$ presented in
Figure~\ref{fig:execution-sample} together with path queries
constructed during the execution.
  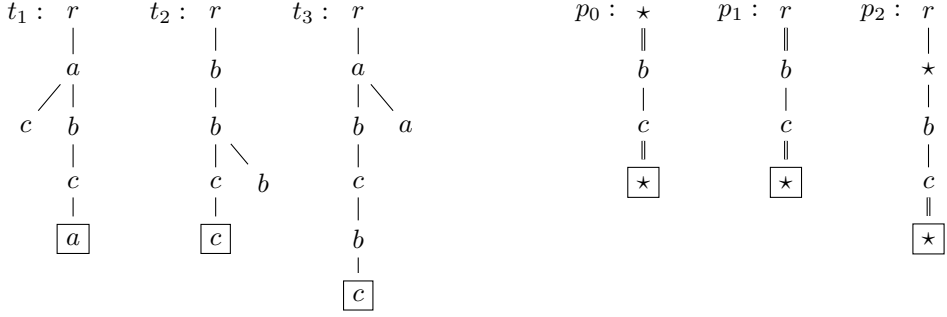
\begin{figure}[htb]
    \centering
    \begin{tikzpicture}[yscale=0.75,xscale=1.25]
      \begin{scope}[xshift=0cm]
        \node at (-0.5,0) (q) {$t_1:$};
        \node at (0,0) (r) {$r$};
        \node at (0,-1) (n1) {$a$}     edge[-] (r);
        \node at (0,-2) (n2) {$b$}     edge[-] (n1);
        \node at (-0.5,-2) (n2') {$c$}     edge[-] (n1);
        \node at (0,-3) (n3) {$c$}     edge[-] (n2);
        \node[draw,outer sep = 3pt] at (0,-4) (n4) {$a$}     edge[-] (n3);
      \end{scope}

      \begin{scope}[xshift=1.5cm]
        \node at (-0.5,0) (q) {$t_2:$};
        \node at (0,0) (r) {$r$};
        \node at (0,-1) (n1) {$b$}     edge[-] (r);
        \node at (0,-2) (n2) {$b$}     edge[-] (n1);
        \node at (0,-3) (n3) {$c$}     edge[-] (n2);
        \node at (0.5,-3) (n3') {$b$}     edge[-] (n2);
        \node[draw,outer sep = 3pt] at (0,-4) (n4) {$c$}     edge[-] (n3);
      \end{scope}
      \begin{scope}[xshift=3cm]
        \node at (-0.5,0) (q) {$t_3:$};
        \node at (0,0) (r) {$r$};
        \node at (0,-1) (n1) {$a$}     edge[-] (r);
        \node at (0,-2) (n2) {$b$}     edge[-] (n1);
        \node at (0.5,-2) (n2') {$a$}     edge[-] (n1);
        \node at (0,-3) (n3) {$c$}     edge[-] (n2);
        \node at (0,-4) (n4) {$b$}     edge[-] (n3);
        \node[draw,outer sep = 3pt] at (0,-5) (n5) {$c$}     edge[-] (n4);
      \end{scope}
      
      \begin{scope}[xshift=6cm]
        \node at (-0.5,0) (q) {$p_0:$};
        \node at (0,0) (r) {$\wc$};
        \node at (0,-1) (n2) {$b$}     edge[-,double] (r);
        \node at (0,-2) (n3) {$c$}     edge[-] (n2);
        \node[draw,outer sep = 3pt] at (0,-3) (n4) {$\wc$}     edge[-,double] (n3);
      \end{scope}

      \begin{scope}[xshift=7.5cm]
        \node at (-0.5,0) (q) {$p_1:$};
        \node at (0,0) (r) {$r$};
        \node at (0,-1) (n2) {$b$}     edge[-,double] (r);
        \node at (0,-2) (n3) {$c$}     edge[-] (n2);
        \node[draw,outer sep = 3pt] at (0,-3) (n4) {$\wc$}     edge[-,double] (n3);
      \end{scope}

      \begin{scope}[xshift=9cm]
        \node at (-0.5,0) (q) {$p_2:$};
        \node at (0,0) (r) {$r$};
        \node at (0,-1) (n1) {$\wc$}     edge[-] (r);
        \node at (0,-2) (n2) {$b$}     edge[-] (n1);
        \node at (0,-3) (n3) {$c$}     edge[-] (n2);
        \node[draw,outer sep = 3pt] at (0,-4) (n4) {$\wc$}     edge[-,double] (n3);
      \end{scope}
    \end{tikzpicture}
    \caption{A sample and the constructed queries.}
    \label{fig:execution-sample}
  \end{figure}

  In the first stage the algorithm identifies a factor $b/c$ present
  in every selecting path and the resulting path query is
  $p_0=\wc\dblslash{}b/c\dblslash{}\wc$. There is no other common
  factor and the algorithm moves to the second stage where it
  specializes the root node of $p_0$ obtaining this way
  $p_1=r\dblslash{}b/c\dblslash{}\wc$; the selecting node cannot be
  specialized because the selected nodes of $t_1$ and $t_2$ have two
  different labels, $a$ and $c$ resp. Finally, the algorithm attempts
  to specialize the descending edges. Only the top one can be replaced
  by a $\wc$-path of length $1$, yielding
  $p_2=r/\wc/b/c\dblslash{}\wc$, which is also the final result of the
  learning algorithm. \qed
\end{example}

There are aspects of the algorithm that are not fully specified e.g.,
from two different subpaths of the same length which one should be
chosen first in the loop in line 4. We do not enforce any particular
choice because it is inessential from the theoretical point (soundness
and completeness) and in practical implementations the choice could be
made with the help of heuristics.
\begin{example}
\label{ex:execution-sample-second}
Consider the sample consisting of the two trees: $r(a(b(c(d))))$ and
$r(b(c(a(b(d)))))$.
In the first stage the algorithm may identify either a factor $a/b$ or
$b/c$ but not both of them. As a consequence the algorithm may return
one of two possible queries $p_1=r\dblslash{}ab\dblslash{}d$ or
$p_2=r\dblslash{}bc\dblslash{}d$.  In order to make the algorithm
deterministic we may enforce some order of processing among candidate
factors of the same length e.g. from left to right. \qed
\end{example}

We observe that $S\subseteq \mathcal{L}_1(p)$ is an invariant
maintained throughout $\learner\APATH_1$ and with a simple analysis
one can show that $\learner\APATH_1$ is sound for $\APATH_1$. But what
makes this algorithm particularly interesting is the following.
\begin{lemma}
  \label{lemma:minimal-learner1}
  The algorithm $\learner\APATH_1$ returns a minimal anchored path
  query consistent with the input sample.
\end{lemma}
We prove the claim below, which by \Prop1 for $\APath_1$ is equivalent
to the lemma above.
\begin{claim}
  If $\learner\APATH_1(S)$ returns $p$, then there is no unary
  anchored path query $q\neq p$ such that $q\preccurlyeq p$ and
  $S\subseteq\mathcal{L}_1(q)$.
\end{claim}
\begin{proof}
  Suppose otherwise and take a unary anchored path query $q\neq p$
  having an embedding $\lambda:p\hookrightarrow q$. We note that $q$
  can be viewed as result of applying a substitution $\theta$ i.e.,
  $q=p\theta$, which substitutes in $p$ some the labels of $\wc$-nodes
  with labels in $\Sigma$ and replaces some of the $\dblslash$-edges
  with path queries. This substitution can be decomposed into a
  composition $\theta=\theta_1\circ\theta_2\circ\ldots\circ\theta_k$
  of atomic operations, which for brevity we present here as rewriting
  rules: 1) $\dblslash{}\mapsto/$ replacing a $\dblslash$-edge with a
  child edge, 2) $\dblslash{}\mapsto\dblslash{}B\dblslash{}$ replacing
  a $\dblslash$-edge with a block $B$ (cf. proof of
  Claim~\ref{claim:1}), 3) $\wc\mapsto a$ changing the $\wc$ label of
  a node to some $a\in\Sigma$, 4)
  $\dblslash{}\mapsto/\wc/\wc/\ldots/\wc/$ replacing a
  $\dblslash$-edge with a $\wc$-path. Now, if we take the path query
  $q'=p\theta_1$, then $q\preccurlyeq q' \preccurlyeq p$ and $q'\neq
  p$. Consequently, it suffices to assume that $p$ is obtained from
  $q$ by applying just one atomic substitution $\theta^\circ$.

  Essentially, the first three types of atomic operations allow to
  identify a new factor or a longer factor that would have been
  discovered and properly incorporated into the resulting query during
  the execution of $\learner\APATH_1(S)$ in lines 4-6. The last type,
  $\dblslash{}\mapsto/\wc/\ldots/\wc$ allows to identify a descending edge that
  would have been converted to a $\wc$-path in lines 13-17. These
  arguments show that $p$ could not have been the result of
  $\learner\APATH_1(S)$; a contradiction.\qed
\end{proof}

We argue that Lemmas~\ref{lemma:prop-anchored-paths}
and~\ref{lemma:minimal-learner1} imply completeness of
$\learner\APATH_1$ w.r.t.\ $\APATH_1$. Indeed, if $\CS_q\subseteq S$
and $\learner\APATH_1(S)$ returns $p$, then $q\subseteq p$ because
$\CS_q\subseteq\mathcal{L}_1(p)$ but there is no query $q'\subset p$
that $S\subseteq\mathcal{L}(q')$. Hence, $q$ and $p$ are equivalent.
\begin{theorem}
  \label{thm:apath1-learnable}
  Anchored path queries are learnable in polynomial time and data from
  positive examples \textnormal{(}i.e., in the setting $\APATH_1$\textnormal{)}.
\end{theorem}

\section{Learning Boolean path queries}
\label{sec:simple-path-patterns}
Learning Boolean path queries is more challenging than learning unary
path queries. In decorated trees, which are examples for learning
unary queries, the selected nodes unambiguously indicate the path to
be matched by the query. The examples for learning Boolean path query
are trees with no indication of the path the constructed query should
match. To address this problem we devise an algorithm that infers a
conjunction of Boolean anchored path queries that are satisfied in the
given sample. Recall that $\APath_0$ is the class of Boolean anchored
path queries and $\CPath_0$ is the class of reduced and
head-consistent conjunctions of Boolean anchored path queries
(represented as sets of Boolean path queries). The corresponding
learning settings are $\APATH_0=(\Tree_0,\APath_0,\mathcal{L}_0)$ and
$\CPATH_0=(\Tree_0,\CPath_0,\mathcal{L}_0)$, where $\mathcal{L}_0$
interprets a set of path queries $P$ as a the twig query obtained by
gluing the root nodes together.

Figure~\ref{fig:path-pattern-inference} contains the learning
algorithms for $\CPATH_0$ and $\APATH_0$. First, we introduce
$\learner\APATH_0^\ast$, a helper learner derived from
$\learner\APATH_1$, which infers a minimal Boolean anchored path query
that is satisfied by the given path $u$ and every tree in the input
sample. Note that to ensure that the output is a Boolean anchored
query $\learner\APATH_0^\ast$ skips the specialization of the last
$\dblslash{}$-edge if doing so would yield a query that is not
anchored (i.e., ending with $\wc$ not preceded immediately by
$\dblslash{}$). The purpose of taking the initial path $u$ from the
input is the ability to consider every path in $S$ as the word in
which to search for common factors.
\begin{figure}[htb]
  \resetLineNoCounter
  \begin{tabbing}
    xxx\=xx\=xx\=xx\=xx\=xx\=xx\=xx\=mmm\kill
    {\bf algorithm} $\learner\APATH_0^\ast(u,S)$\\
    {\bf Input:} a path $u$ and a sample $S\subseteq\Tree_0$ of trees\\
    {\bf Output:} a minimal $p\in\APath_0$  
    s.t. $S\cup\{u\}\subseteq\mathcal{L}_0(p)$\\
    This algorithm is obtained from $\learner\APATH_1$ by: \\
    \> $\bullet$ initializing $w$ to $u$ (line~1) \\
    \> $\bullet$ replacing every
    $S\subseteq\mathcal{L}_1(p)$
    by $S\cup\{w\}\subseteq \mathcal{L}_0(p)$\\
    \> $\bullet$ skipping the execution of loop 13--17 for \\
    \> \> the last $\dblslash{}$-edge if $b_1=\wc$.\\[15pt]
    {\bf algorithm} $\learner\CPATH_0(S)$\\
    {\bf Input:} a sample $S\subseteq\Tree_0$ of trees\\
    {\bf Output:} a set of minimal queries 
    $P\subseteq\APath_0$\\
    \>\>\>\>such that $S\subseteq\mathcal{L}_0(P)$\\
    \lineNo \> $P\colonequals\emptyset$\\
    \lineNo \> {\bf for} $u\in\Paths(S)$ {\bf do}\\
    \lineNo \> \> $p\colonequals\learner\APATH_0^\ast(u,S)$\\
    \lineNo \> \> {\bf if} $\nexists q\in P.\ q\preccurlyeq p$ {\bf then}\\
    \lineNo \> \> \> $P\colonequals P\minus \{ q\in P\mid p \preccurlyeq q\}$\\
    \lineNo \> \> \> $P\colonequals P\cup\{p\}$\\
    \lineNo \> {\bf return} $P$\\[15pt]
    \resetLineNoCounter
    {\bf algorithm} $\learner\APATH_0(S)$\\
    {\bf Input:} a sample $S\subseteq\Tree_0$ of trees\\
    {\bf Output:} a minimal $p\in\APath_0$ 
    such that $S \subseteq \mathcal{L}_0(p)$\\
    \lineNo \> $P\colonequals\learner\CPATH_0(S)$\\
    \lineNo \> choose any $p$ from $P$\\
    \lineNo \> {\bf return} $p$
  \end{tabbing}
  \caption{Learning $\CPATH_0$ and $\APATH_0$.}
  \label{fig:path-pattern-inference}
\end{figure}

Essentially, $\learner\CPATH_0$ considers every path $u$ in tree of
$S$ and uses $\learner\APATH_0^\ast$ to find a most specific (i.e.,
minimal) Boolean path query $p$ satisfied by $u$ and every other
element of $S$. The set $P$ aggregates all minimal results of running
$\learner\APATH_0^\ast$ over all paths in the input sample. The
learning algorithm $\learner\APATH_0$ simply takes the result of
$\learner\CPATH_0$ and chooses one element. The choice is arbitrary,
but later, we show that in the presence of the characteristic sample
$\learner\CPATH$ returns a singleton and there is no ambiguity.
\begin{example}
  \label{ex:2}
  We run $\learner\CPATH_0$ on the sample $S_0$
  (Figure~\ref{fig:sample-for-cpath}) 
  \begin{figure}[htb]
  \centering
  \begin{tikzpicture}[yscale=0.90,xscale=1.25]
    \begin{scope}[xshift=0cm]
      \node (n0) at (0,0) {\tt offer};
      \node (n1) at (0,-1) {\tt item} edge[-] (n0);
      \node (n2) at (-0.5,-2) {\tt for-sale} edge[-] (n1);
      \node (n4) at (0.5,-2) {\tt descr} edge[-] (n1);
    \end{scope}
    \begin{scope}[xshift=3.5cm]
      \node (n0) at (0,0) {\tt offer};
      \node (n0') at (0,-0.75) {\tt list} edge[-] (n0);
      \node (n1) at (-1,-1.33) {\tt item} edge[-] (n0');
      \node (n2) at (-1.5,-2.17) {\tt for-sale} edge[-] (n1);
      \node (n4) at (-0.5,-2.17) {\tt descr} edge[-] (n1);
      \node (n1') at (1,-1.33) {\tt item} edge[-] (n0');
      \node (n2') at (0.5,-2.17) {\tt wanted} edge[-] (n1');
      \node (n4') at (1.5,-2.17) {\tt descr} edge[-] (n1');
    \end{scope}
  \end{tikzpicture}
  \caption{\label{fig:sample-for-cpath}Input sample from
    Example~\ref{ex:2}.}
\end{figure}
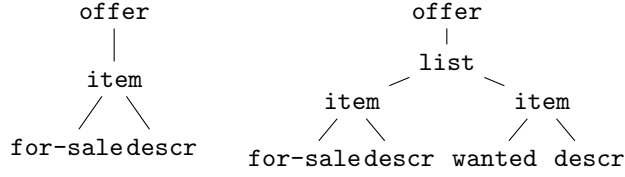
corresponding to the positive examples from
Example~\ref{ex:intro-boolean} simplified for clarity of
presentation. The set of paths $\Paths(S_0)$ in the sample consists
of:
\begin{align*} 
&u_1=\mathtt{offer/item/for\text{-}sale},\\
&u_2=\mathtt{offer/item/descr},\\
&u_3=\mathtt{offer/list/item/for\text{-}sale},\\
&u_4=\mathtt{offer/list/item/descr}, \\
&u_5=\mathtt{offer/list/item/wanted}.
\end{align*}
Running $\learner\APATH_0^\ast$ on those paths yields:
\begin{align*}
&\learner\APATH_0^\ast(u_1,S_0)=\mathtt{offer\dblslash{}item/for\text{-}sale},\\
&\learner\APATH_0^\ast(u_2,S_0)=\mathtt{offer\dblslash{}item/descr},\\
&\learner\APATH_0^\ast(u_3,S_0)=\mathtt{offer\dblslash{}item/for\text{-}sale},\\
&\learner\APATH_0^\ast(u_4,S_0)=\mathtt{offer\dblslash{}item/descr}, \\
&\learner\APATH_0^\ast(u_5,S_0)=\mathtt{offer\dblslash{}item\dblslash{}\wc}.  
\end{align*}
Note that the result of $\learner\APATH_0^\ast$ on $u_5$ is the
Boolean anchored query $\mathtt{offer\dblslash{}item\dblslash{}\wc}$
and not the more specific $\mathtt{offer\dblslash{}item/\wc}$ because
it is not anchored; $\learner\APATH_0^\ast$ skips the attempt to
specialize the last $\dblslash{}$-edge because it is followed by
$\wc$. The query $\mathtt{offer\dblslash{}item\dblslash{}\wc}$ is,
however, subsumed by all the previous queries, and therefore,
$\learner\CPATH_0(S_0)$ returns a set containing only the queries
$\mathtt{offer\dblslash{}item/for\text{-}sale}$ and
$\mathtt{offer\dblslash{}item/descr}$. The run of $\learner\APATH_0$
on $S_0$ returns one of those queries e.g., the one whose string
representation is lexicographically minimal
$\mathtt{offer\dblslash{}item/descr}$. While this is not best choice
for Example~\ref{ex:intro-boolean}, the negative examples can be used
in a heuristic to select a query rejecting the most negative examples,
in this case $\mathtt{offer\dblslash{}item/for\text{-}sale}$. \qed
\end{example}

Because $\learner\CPATH_0(S)$ returns a set $P$ of Boolean path
queries that are satisfied in every tree in $S$, this algorithm is
sound. Naturally, $\learner\APATH_0$ is also sound because it returns
one element of $P$. To show completeness of both learning algorithms,
we point out an important property of $\learner\APATH_0^\ast$. The
construction of characteristic samples $\CS_P$ and $\CS_p$ is in
Section~\ref{sec:match-sets-char}.
\begin{lemma}
  \label{lemma:learner*}
  Take a conjunctive query $P\in\CPath_0$, let $\CS_P=\{t_0,t_1\}$ be
  the characteristic sample for $P$, and take any sample
  $S\subseteq\mathcal{L}_0(P)$ containing two examples $t_0',t_1'$
  such that $\Paths(t_i)=\Paths(t_i')$ for $i\in\{0,1\}$. Then,
  \begin{enumerate}
    \itemsep0pt
  \item for every $u\in\Paths(S)$, $\learner\APATH_0^\ast(u,S)$
    returns a path query equal to or subsumed by some $p\in P$.
  \item for every $p\in P$ there exists $u\in\Paths(S)$ such that
    $\learner\APATH_0^\ast(u,S)$ returns $p$.
  \end{enumerate}
\end{lemma}
The above result shows completeness of $\learner\CPATH_0$. As for
$\learner\APATH_0$, if we take a Boolean anchored path query $p$ and
apply the previous lemma to $P=\{p\}$, we get that for any sample $S$
consistent with $p$ and containing $\CS_p$ the algorithm
$\learner\CPATH_0(S)$ returns the singleton $\{p\}$, and thus,
$\learner\APATH_0(S)$ returns $p$. This result allows to prove
learnability of both classes of queries.
\begin{theorem}
  The query classes $\CPath_0$ and $\APath_0$ are learnable in
  polynomial time and data from positive examples \textnormal{(}i.e.,
  in the settings $\CPATH_0$ and $\APATH_0$ resp.\textnormal{)}
\end{theorem}
We also show minimality of $\learner\APATH_0$.
\begin{lemma}
  \label{lemma:learner2-and-3-minimal}
  For any finite $S\subseteq\Tree_0$, $\learner\CPATH_0(S)$
  returns a set of minimal Boolean anchored path queries consistent
  with $S$ and $\learner\APATH_0(S)$ returns a minimal Boolean
  anchored path query consistent with $S$.
\end{lemma}
We point out that while the result of $\learner\CPATH_0(S)$ is a set
of minimal queries, it is not necessarily a minimal conjunctive query
i.e., it is not a maximal set of minimal queries. In the example below
we show that a set of positive examples may have an exponential number
of minimal Boolean path queries, and therefore, constructing their
conjunction cannot be done in polynomial time.
\begin{example}
  \label{ex:exp-minimal-queries}
  Fix $n>0$ and take the set of positive examples $S_{\mathrm{exp}}$
  of containing exactly two trees
  \begin{align*}
  &t_0=r(a_1(b_1(\ldots a_n(b_n(c))\ldots))),\\
  &t_1=r(b_1(a_1(\ldots b_n(a_n(c))\ldots))).
  \end{align*}
  Any query of the form $r\dblslash{}\beta_1\dblslash{}\ldots\dblslash{}\beta_n\dblslash{}c$, with
  $\beta_i\in\{a_i,b_i\}$ and $i\in\{1,\ldots,n\}$, is a minimal
  Boolean path query consistent with $S_\mathrm{exp}$.\qed
\end{example}

\section{Learning Boolean twig queries}
\label{sec:tree-patterns}

It this section we investigate learning path-subsumption-free twig
queries from positive examples i.e., the learning setting
$\PTWIG_0=(\Tree_0,\PTwig_0,\mathcal{L}_0)$. Recall that
$q\in\PTwig_0$ is query such that the set of root-to-leaf paths
$\Paths(q)$ consists of Boolean anchored path queries and does not
contain two path queries such that one subsumes another. Our approach
is based on the algorithm $\learner\CPATH_0$, which infers a set $P$
of minimal Boolean path queries and a method that allows to
reconstruct a twig query from path queries in $P$. Intuitively
speaking, we shall interleave the path queries from $P$ to obtain the
twig query. Below, we describe formally this technique.

Given a path query $p$ and a node $n\in N_p$, the \emph{split} of $p$
at $n$ is a pair of path queries $p_1$ and $p_2$ such that $p_1$ is
the path from $\root_p$ to $n$ and $p_2$ is the path from $n$ to the
only leaf of $p$. Note that $n$ becomes the root node of $p_2$. A
\emph{fusion} of $p$ into a twig query $q$ is a twig query $q'$ such
that the pair $p_1$ and $p_2$ is a split of $p$ at $n$, there exists
an embedding $\lambda:p_1\hookrightarrow q$, and $q'$ is obtained from
$q$ by attaching $p_2$ at node $\lambda(n)$ (the node $\lambda(n)$ and
the root node $n$ of $p_2$ become the same node, the label of $n$ in
$p_2$ is ignored). By $\Fusions(p,q)$ we denote the set of all fusions
of $p$ into $q$. Figure~\ref{fig:fusions} presents all fusions of
$r\dblslash{}a/b$ into $r[\wc/a]\dblslash{}a/c$.
\begin{figure}[htb]
  \centering
  \begin{tikzpicture}[yscale=0.75,xscale=1.25]
      \begin{scope}[xshift=0cm]
        \node at (-0.5,0) (q) {$p_0:$};
        \node at (0, 0) (m0) {$r$};
        \node at (0,-1) (m1) {$a$} edge[-,double] (m0);
        \node at (0,-2) (m2) {$b$} edge[-] (m1);
      \end{scope}
      \begin{scope}[xshift=1.25cm]
        \node at (-0.5,0) (q) {$q_0:$};
        \node at (0,0) (n0) {$r$};
        \node at (-0.5,-1) (n1) {$\wc$} edge[-] (n0);
        \node at (0.5,-1) (n2) {$a$} edge[-,double] (n0);
        \node at (-0.5,-2) (n3) {$a$} edge[-] (n1);
        \node at (0.5,-2) (n4) {$c$} edge[-] (n2);
      \end{scope}
      
      \draw (2.5,-1) edge[-latex] node[above] {$\Fusions$} (4,-1);

      \begin{scope}[xshift=5.5cm]
        \node at (-0.5,0) (q) {$q_1:$};
        \node at (0,0) (n0) {$r$};
        \node at (-0.75,-1) (n1) {$\wc$} edge[-] (n0);
        \node at (0.75,-1) (n2) {$a$} edge[-,double] (n0);
        \node at (-0.75,-2) (n3) {$a$} edge[-] (n1);
        \node at (0.75,-2) (n4) {$c$} edge[-] (n2);
        \node at (0,-1) (m1) {$a$} edge[-,double] (n0);
        \node at (0,-2) (m2) {$b$} edge[-] (m1);
      \end{scope}

      \begin{scope}[xshift=7.5cm]
        \node at (-0.5,0) (q) {$q_2:$};
        \node at (0,0) (n0) {$r$};
        \node at (-0.5,-1) (n1) {$\wc$} edge[-] (n0);
        \node at (0.5,-1) (n2) {$a$} edge[-,double] (n0);
        \node at (-0.5,-2) (n3) {$a$} edge[-] (n1);
        \node at (0.5,-2) (n4) {$c$} edge[-] (n2);
        \node at (-0.5,-3) (m2) {$b$} edge[-] (n3);
      \end{scope}

      \begin{scope}[xshift=9.25cm]
        \node at (-0.5,0) (q) {$q_3:$};
        \node at (0,0) (n0) {$r$};
        \node at (-0.5,-1) (n1) {$\wc$} edge[-] (n0);
        \node at (0.5,-1) (n2) {$a$} edge[-,double] (n0);
        \node at (-0.5,-2) (n3) {$a$} edge[-] (n1);
        \node at (0.25,-2) (n4) {$c$} edge[-] (n2);
        \node at (0.75,-2) (m2) {$b$} edge[-] (n2);
      \end{scope}

    \end{tikzpicture}
  \caption{Fusions of $p_0$ into $q_0$.}
  \label{fig:fusions}
\end{figure}
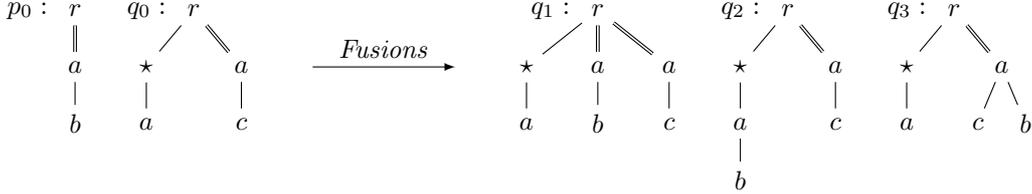

We point out that if $q$ is path-subsumption-free and $p$ is anchored,
then all elements of $\Fusions(p,q)$ are path-sub\-sump\-tion-free. We
note that $\Fusions(p,q)$ may be empty e.g., there is no fusion of
$a/a$ into $b[a]/b$, but as we argue next, this is never the case in
the learning algorithm $\learner\PTWIG_0$ which we present in
Figure~\ref{fig:tree-pattern-inference}. We slightly extend the
notation: $\varnothing$ denotes a \emph{phantom} empty twig query and
$\Fusions(\varnothing,p)=\{p\}$.

\begin{figure}[htb]
  \resetLineNoCounter
  \begin{tabbing}
    xxx\=xx\=xx\=xx\=xx\=xx\=xx\=xx\=mmm\kill
    {\bf algorithm} $\learner\PTWIG_0(S)$ \\
    {\bf Input:} a sample $S\subseteq\Tree_0$ of trees\\
    {\bf Output:} a query $p\in\PTwig_0$ such that $S \subseteq \mathcal{L}_0(p)$\\
    \lineNo \> $q\colonequals\varnothing$\\
    \lineNo \> $P\colonequals\learner\CPATH_0(S)$\\
    \lineNo \> {\bf for} $p\in P$ {\bf do}\\
    \lineNo \> \> $C\colonequals\{q'\in\Fusions(p,q)\mid S\subseteq \mathcal{L}_0(q')\}$\\
    \lineNo \> \> $q\colonequals\text{choose any $\preccurlyeq$-minimal element of $C$}$\\
    \lineNo \> {\bf return} $q$
  \end{tabbing}
  \caption{Learning algorithm for $\PTWIG_0$.}
  \label{fig:tree-pattern-inference}
\end{figure}

Basically, $\learner\PTWIG_0$ uses $\learner\CPATH_0$ to construct a
set $P$ of Boolean path queries satisfied in all trees of $S$ and then
fusions all the paths into one twig query. Note that $C$ is never
empty because $q$ is build up from path queries in $P$ that are
satisfied in $S$ and have the same label in their root
nodes. Consequently, $\learner\PTWIG_0$ executes without errors and is
sound. The order in which $\learner\PTWIG_0$ performs fusions is
arbitrary, but later on, we show that in the presence of the
characteristic sample, the set $C$ has exactly one element at all
times, and the final result is the goal query. First, we illustrate
the work of $\learner\PTWIG_0$ on an example.
\begin{example}
  \label{ex:learnerCPath0}
  Consider a sample $S_1$ containing two DBLP listings in
  Figure~\ref{fig:dblp-data-set}: one with a collection of articles
  and the other with a collection of books.
  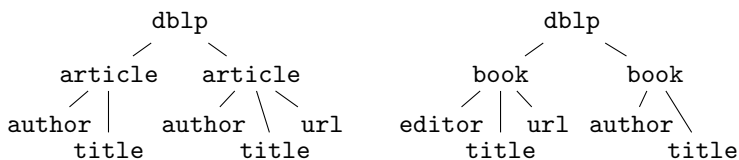
\begin{figure}[htb]
  \centering
  \begin{tikzpicture}[yscale=0.9,xscale=1.25]
    \begin{scope}[xshift=0cm,xscale=0.5,yscale=0.75]
      \node (n0) at (0,0) {\tt dblp};
      \node (n1) at (-1.5,-1) {\tt article} edge[-] (n0);
      \node (n3) at (-2.75,-2) {\tt author} edge[-] (n1);
      \node (n2) at (-1.5,-2.5) {\tt title} edge[-] (n1);
      \node (n4) at (1.5,-1) {\tt article} edge[-] (n0);
      \node (n5) at (0.5,-2) {\tt author} edge[-] (n4);
      \node (n6) at (2,-2.5) {\tt title} edge[-] (n4);
      \node (n7) at (3,-2) {\tt url} edge[-] (n4);
    \end{scope}
    \begin{scope}[xshift=4.125cm,xscale=0.5,yscale=0.75]
      \node (n0) at (0,0) {\tt dblp};
      \node (n1) at (-1.5,-1) {\tt book} edge[-] (n0);
      \node (n3) at (-2.75,-2) {\tt editor} edge[-] (n1);
      \node (n2) at (-1.5,-2.5) {\tt title} edge[-] (n1);
      \node (n7) at (-0.5,-2) {\tt url} edge[-] (n1);
      \node (n4) at (1.75,-1) {\tt book} edge[-] (n0);
      \node (n5) at (1.25,-2) {\tt author} edge[-] (n4);
      \node (n6) at (2.75,-2.5) {\tt title} edge[-] (n4);
    \end{scope}
  \end{tikzpicture}
  \caption{\label{fig:dblp-data-set} Input sample}
\end{figure}

$\learner\CPATH_0(S_1)$ returns the following path
queries:
\begin{small}
\begin{align*}
&p_1=\mathtt{dblp/\wc/author},&
&p_2=\mathtt{dblp/\wc/title},&
&p_3=\mathtt{dblp/\wc/url}.
\end{align*}%
\end{small}%
We perform fusions in the order $p_1$, $p_2$, and $p_3$. Fusing $p_1$
and $p_2$ yields the query $\mathtt{dblp/\wc[title]/author}$ and
fusing $p_3$ into it gives
$q'=\mathtt{dblp[\wc/url]/\wc[title]/author}$. Note that in the last
step, $\mathtt{dblp/\wc[title][url]/author}$ is one of the fusions but
it is not consistent with the input sample $S_1$. On the other hand,
if the order of fusions is $p_2$, $p_3$, and $p_1$, then the end
result is $q''=\mathtt{dblp[\wc/author]/\wc[title]/url}$.  \qed
\end{example}
While in the previous example the queries $q'$ and $q''$ are minimal
path-subsumption-free twig queries consistent with $S_1$, in general
$\learner\PTWIG_0$ does not need to produce such minimal queries. In
fact, we show that for certain samples, such a minimal query may be of
exponential size and thus impossible to construct by a polynomial
algorithm.
\begin{example}[cont'd Example~\ref{ex:exp-minimal-queries}]
  \label{ex:1}
  Recall the sample $S_{\mathrm{exp}}$ and observe that the minimal
  twig query consistent with $S_{\mathrm{exp}}$ has the shape of a
  perfect binary tree of height $n+1$ where every node at depth
  $i\in\{0,\ldots,n-1\}$ has two children labeled with $a_{i+1}$ and
  $b_{i+1}$ (connected with their parent with a
  $\dblslash$-edge). Naturally, this minimal query is
  path-subsumption-free.\qed
\end{example}
Now, we move to completeness of $\learner\PTWIG_0$ and we fix a query
$q\in\PTwig_0$ and a sample $S\subseteq \mathcal{L}_0(q)$. Recall the
construction of the characteristic sample $\CS_q$ for $q$ from
Section~\ref{sec:match-sets-char}. First, we observe that for
$q\in\PTwig_0$ every $p\in\Paths(q)$ is a $\preccurlyeq$-minimal
element of $\Paths(q)$. As a simple consequence of
Lemma~\ref{lemma:learner*} we get the following.
\begin{lemma}
  \label{lemma:paths-of-a-twig}
  If $S$ contains $\CS_q$, then $\learner\CPATH_0(S)$ returns
  $\Paths(q)$.
\end{lemma}
To state that the algorithm approaches the goal query $q$ with every
fusion, we need to define formally the search space of subqueries of
$q$ and show that when moving with the fusion operator we never leave
the space and finally reach $q$. A Boolean twig query $q'$ is a
\emph{subquery} of $q$ if there exists a subset $N$ of leaves of $q$
such that $q'$ is a subgraph induced by the set of paths from the root
of $q$ to the leaves in $N$. The main claim follows.
\begin{lemma}
  Assume that $\CS_q\subseteq S$. For any subquery $q'$ of $q$, and
  any path query $p\in\Paths(q)\minus\Paths(q')$ the set of elements
  of $\Fusions(p,q')$ consistent with $S$ has exactly one
  $\preccurlyeq$-minimal element $q''$. Furthermore, $q''$ is a
  subquery of $q$.
\end{lemma}
If $\CS_q\subseteq S$, then by Lemma~\ref{lemma:paths-of-a-twig}
$P=\Paths(q)$, and therefore, whatever is the order of choosing paths
from $P$ in line 3, the algorithm $\learner\PTWIG_0$ approaches $q$
and when all paths in $P$ are fused, we obtain $q$.
\begin{theorem}
  \label{thm:twig-learnable}
  Path-subsumption-free Boolean twig queries are learnable in
  polynomial time and data from positive examples \textnormal{(}i.e.,
  in the setting $\PTWIG_0$\textnormal{)}.
\end{theorem}

\section{Learning unary twig queries}
\label{sec:tree-queries}
In this section, we present an algorithm $\learner\PTWIG_1$
(Figure~\ref{fig:learner-psf-twig1}) for learning unary
path-subsumption-free twig queries from positive examples i.e., in the
learning setting $\PTWIG_1=(\Tree_1,\PTwig_1,\mathcal{L}_1)$.

Essentially, the learning algorithm uses $\learner\APATH_1$ to
construct a path query $p$ and then it uses $\learner\PTWIG_1^\ast$, a
helper learner derived from $\learner\PTWIG_0$, to decorate the nodes
of the path query with filter expressions (Boolean twig
queries). Here, we use the non-abbreviated syntax of XPath to
represent the path query $p$ as
$\ell_0/\alpha_1\sep\ell_1/\ldots/\alpha_k\sep\ell_k$, where
$\ell_i\in\Sigma\cup\{\wc\}$ and $\alpha_i$ is either $\mathit{child}$
or $\mathit{descendant}$.

\begin{figure}[htb]
  \centering
  \resetLineNoCounter
  \begin{tabbing}
    xxx\=xx\=xx\=xx\=xx\=xx\=xx\=xx\=mmm\kill
    {\bf algorithm} $\learner\PTWIG_1^\ast(S,q')$ \\
    {\bf Input:} a sample $S\subseteq\Tree_1$ of decorated trees and \\
    \>\>\> a query $q'\in\PTwig_1$ such that 
    $S\subseteq\mathcal{L}_1(q')$ \\
    {\bf Output:} a query $q\in\PTwig_1$ s.t.\ 
    $q\preccurlyeq q'$ and $S \subseteq \mathcal{L}_1(q)$\\
    This algorithm is obtained from $\learner\PTWIG_0$ by:\\
    \> $\bullet$ initializing $q$ to $q'$ (line~1)\\
    \> $\bullet$ replacing every $\mathcal{L}_0$ by $\mathcal{L}_1$\\[15pt]
    {\bf algorithm} $\learner\PTWIG_1(S)$ \\
    {\bf Input:} a sample $S$ of decorated trees \\
    {\bf Output:} a query $q\in\PTwig_1$ such that $S\subseteq\mathcal{L}_1(q)$\\
    \lineNo\> $p\colonequals\learner\APATH_1(S)$\\
    \lineNo\> let $p$ be of the form 
    $\ell_0/\alpha_1\sep\ell_1/\ldots/\alpha_k\sep\ell_k$\\ 
    \lineNo\> $q_k'\colonequals\ell_k$ \\
    \lineNo\> {\bf for} $i=k,\ldots,0$ {\bf do}\\
    \lineNo\> \> $S_i\colonequals\emptyset$\\
    \lineNo\> \> {\bf for} $t\in S$ {\bf do}\\
    \lineNo\> \> \> let $n$ be the deepest node on the path from \\
    \> \> \> \> \> the root node $\root_t$ to the selected node $\sel_t$, \\
    \> \> \> \> \> such that $n$ is reachable from $\root_t$ with \\
    \> \> \> \> \> $\ell_0/\alpha_1\sep\ell_1/\ldots/\alpha_i\sep\ell_i$ 
    and $\sel_t$ is reachable  \\ 
    \> \> \> \> \> from $n$ with $q_i'$\\
    \lineNo\> \> \> add the subtree of $t$ rooted at $n$ to $S_i$ \\
    \lineNo\> \> $q_i\colonequals\learner\PTWIG_1^\ast(S_i,q_i')$\\
    \lineNo\> \> {\bf if} $i > 0$ {\bf then}\\
    \lineNo\> \> \> $q_{i-1}'\colonequals\ell_{i-1}/\alpha_i\sep q_i$\\
    \lineNo\> {\bf return $q_0$}
  \end{tabbing}
  \caption{Learning algorithm for $\PTWIG_1$.}
  \label{fig:learner-psf-twig1}
\end{figure}

When decorating the $i$-th step of $p$ i.e., the fragment
$\alpha_i\sep\ell_i$, with a filter expression, the algorithm first
constructs a sample $S_i$ of subtrees that serve as positive examples
for learning the corresponding filter expression. From every decorated
tree in the input sample $S$ one subtree is extracted. Each subtree is
rooted at a node $n$ on the path from the root node to the selected
node of the decorated tree $t$. The choice of $n$ is done so that it
can be reached with the unprocessed part of the path query
$\ell_0/\alpha_1\sep\ell_1/\ldots/\alpha_i\sep\ell_i$ and at the same
time the decorated part of the path query $q_i'$ selects the selected
node $\sel_t$ when evaluated from $n$. An important invariant of the
outer for loop (lines 4-12) is that there is at least one such $n$ for
every $t\in S$. If there is more than one possible choice, the deepest
node is chosen.

\begin{example}
  \label{ex:learner-psf-twig1}
  Consider a sample $S_2$ (Figure~\ref{fig:samples-library}) that
  contains the positive examples corresponding to (a simplified
  version of) the document from Example~\ref{ex:intro-unary}.
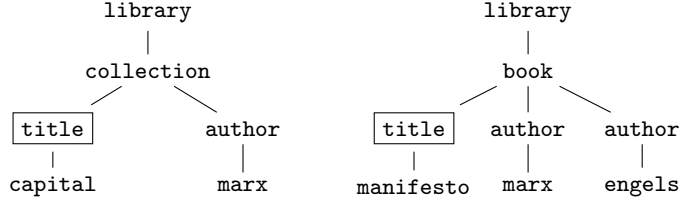
\begin{figure}[htb]
  \centering
  \begin{tikzpicture}[yscale=0.76925,xscale=1.25]\small
    \node (r) at (0,0) {\tt library};
    \node (n1) at (0,-1) {\tt collection} edge[-] (r);
    \node[draw,outer sep = 3pt] (n4) at (-1,-2) {\tt title} edge[-] (n1);
    \node (n5) at (1,-2) {\tt author} edge[-] (n1);
    \node (n11) at (-1,-3) {\tt capital} edge[-] (n4);
    \node (n12) at (1,-3) {\tt marx} edge[-] (n5);
    \begin{scope}[xshift=4cm]
    \node (r) at (0,0) {\tt library};
    \node (n2) at (0,-1) {\tt book} edge[-] (r);

    \node[draw,outer sep = 3pt] (n6) at (-1.2,-2) {\tt title} edge[-] (n2);
    \node (n7) at (-0,-2) {\tt author} edge[-] (n2);
    \node (n8) at (1.2,-2) {\tt author} edge[-] (n2);
    \node (n11) at (-1.2,-3) {\tt manifesto} edge[-] (n6);
    \node (n12) at (-0,-3) {\tt marx} edge[-] (n7);
    \node (n13) at (1.2,-3) {\tt engels} edge[-] (n8);
  \end{scope}
  \end{tikzpicture}
  \caption{Examples from a library database}
  \label{fig:samples-library}
\end{figure}
$\learner\APATH_1(S_2)$ returns the query
$p=\mathtt{/library/\wc/title}$. The algorithm attempts to specialize
the bottom fragment $q_2'=\mathtt{title}$ using the two subtrees
$\mathtt{title(capital)}$ and $\mathtt{title(manifesto)}$. The only
Boolean anchored path query these subtrees do have in common is
$\mathtt{title\dblslash{}\wc}$, which is fused into the query yielding
$q_2=\mathtt{title[.\dblslash{}\wc]}$. Next, the algorithm moves to
$q_1'=\mathtt{\wc/title[.\dblslash{}\wc]}$ and calls $\learner\PTWIG_1^\ast$
with two subtrees: one at the node \texttt{collection} and one rooted
at the node \texttt{book}. $\learner\CPATH_0$ called with these two
trees on input returns two path queries $\mathtt{\wc/title\dblslash{}\wc}$ and
$\mathtt{\wc/author/marx}$. The first path query is subsumed by
$q_1'$, and therefore, it is absorbed by $q_1'$ when fusing. Fusing
the second path query into $q_1'$ yields the query
$q_1=\mathtt{\wc[author/marx]/title[.\dblslash{}\wc]}$. Finally, the algorithm
moves level up to the query
$q_0'=q_0=\mathtt{library/\wc[author/marx]/title[.\dblslash{}\wc]}$, which is
also the end result of $\learner\PTWIG_1$. \qed
\end{example}
We observe that $q_0$ can be considered as overspecialized: it
contains the filter expression $[.\dblslash{}\wc]$ which tests that
the selected \texttt{title} nodes have contents, a test trivially true
in the presence of a reasonable schema information. Currently,
however, our algorithms do not take advantage of schema information.

The soundness of $\learner\PTWIG_1$ follows from the invariant of the
main loop (lines 4--12): for every $t\in S$ in line $7$ there is at
least one node with the desired property. Completeness of
$\learner\PTWIG_1$ follows essentially from completeness of the
algorithms $\learner\APATH_1$ and $\learner\PTWIG_0$, and from the
fact that in line~7 we chose the deepest node.
\begin{theorem}
\label{thm:learning-unary-twig}
  Path-subsumption-free unary twig queries are learnable in polynomial
  time and data from positive examples \textnormal{(}i.e., in the
  setting $\PTWIG_1$\textnormal{)}.
\end{theorem}

\section{Impact of negative examples}
\label{sec:impact-negat-exampl}
In the previous sections, we considered the setting where the user
provides positive examples only. In this section, we allow the user to
additionally specify negative examples. We use two symbols $+$ and $-$
to mark whether an example $t$ of some query is a positive one $(t,+)$
or a negative one $(t,-)$. Formally, for $i\in\{0,1\}$ we consider the
following learning settings:
$\PATH_i^{\pm}=(\Tree_i^\pm,\Path_i,\mathcal{L}_i^\pm)$ and
$\TWIG_i^{\pm}=(\Tree_i^\pm,\Twig_i,\mathcal{L}_i^\pm)$, where
$\Tree_i^\pm=\Tree_i\times\{+,-\}$ and
$\mathcal{L}_i^\pm(q)=\mathcal{L}_i(q)\times\{+\}\cup
(\Tree_i\minus\mathcal{L}_{i}(q))\times\{-\}$.

We study the problem of checking whether there even exists a query
consistent with the input sample because any sound learning algorithm
needs to return $\Null$ if and only if there is no such
query. Formally, given a learning setting
$\mathcal{K}=(\mathcal{D},\mathcal{C},\mathcal{L})$, the {\em
  $\mathcal{K}$-consistency} is the following decision problem
\[
\mathsf{CONS}_{\mathcal{K}} =\{S\subseteq\mathcal{D} \mid \exists
q\in\mathcal{C}.\ S\subseteq \mathcal{L}(q) \}.
\]

Note that in the presence of positive examples the consistency problem
is trivial as long as the query class contains the universal query
$\wc\dblslash{}\wc$. In the presence of negative examples this problem becomes
quite complex.
\begin{theorem}
  $\TWIG_i^\pm$-consistency is NP-complete for any $i\in\{0,1\}$
  (even in the presence of one negative example).
\end{theorem}
\begin{proof}
  We only outline the proof of NP-hardness of
  $\TWIG_0^\pm$-consistency with a reduction from
  $\mathsf{SAT}$. Showing the membership to NP is more difficult,
  uses a nontrivial minimal-witness argument, and is omitted.

  We illustrate the reduction on an example of a $\mathsf{CNF}$
  formula $\varphi_0=(\neg x_1\lor x_2 \lor \neg x_3)\land(x_1\lor
  \neg x_2)$ for which the corresponding sample is presented in
  Figure~\ref{fig:reduction} (positive and negative examples are
  indicated with the symbols $+$ and $-$ respectively).
  \begin{figure*}[htb!]
    \centering
    \begin{tikzpicture}[yscale=0.85,xscale=0.5]
      \fontsize{6}{7}
      \selectfont
      \begin{scope}[yshift=-1cm,xshift=-8.25cm]
        \node (c) at (0,0) {$c$};
        \path (c) node[above=0.2cm] {${+}$};
        \begin{scope}[yshift=-1cm,xshift=-3cm]
          \node (d) at (0,0) {$d$} edge [-,semithick] (c);
          \begin{scope}[yshift=-1cm,xshift=-1cm]
            \node (x) at (0,0) {$x_1$} edge [-,semithick] (d);
            \node (f) at (0,-1) {$0$} edge [-,semithick] (x);
          \end{scope}
          \begin{scope}[yshift=-1cm,xshift=0cm]
            \node (x) at (0,0) {$x_2$} edge [-,semithick] (d);
            \node (f) at (-.25,-1) {$0$} edge [-,semithick] (x);
            \node (t) at (.25,-1) {$1$} edge [-,semithick] (x);
          \end{scope}
          \begin{scope}[yshift=-1cm,xshift=1cm]
            \node (x) at (0,0) {$x_3$} edge [-,semithick] (d);
            \node (f) at (-.25,-1) {$0$} edge [-,semithick] (x);
            \node (t) at (.25,-1) {$1$} edge [-,semithick] (x);
          \end{scope}
        \end{scope}
        \begin{scope}[yshift=-1cm,xshift=0cm]
          \node (d) at (0,0) {$d$} edge [-,semithick] (c);
          \begin{scope}[yshift=-1cm,xshift=-1cm]
            \node (x) at (0,0) {$x_1$} edge [-,semithick] (d);
            \node (f) at (-.25,-1) {$0$} edge [-,semithick] (x);
            \node (t) at (.25,-1) {$1$} edge [-,semithick] (x);
          \end{scope}
          \begin{scope}[yshift=-1cm,xshift=0cm]
            \node (x) at (0,0) {$x_2$} edge [-,semithick] (d);
            \node (t) at (0,-1) {$1$} edge [-,semithick] (x);
          \end{scope}
          \begin{scope}[yshift=-1cm,xshift=1cm]
            \node (x) at (0,0) {$x_3$} edge [-,semithick] (d);
            \node (f) at (-.25,-1) {$0$} edge [-,semithick] (x);
            \node (t) at (.25,-1) {$1$} edge [-,semithick] (x);
          \end{scope}
        \end{scope}
        \begin{scope}[yshift=-1cm,xshift=3cm]
          \node (d) at (0,0) {$d$} edge [-,semithick] (c);
          \begin{scope}[yshift=-1cm,xshift=-1cm]
            \node (x) at (0,0) {$x_1$} edge [-,semithick] (d);
            \node (f) at (-.25,-1) {$0$} edge [-,semithick] (x);
            \node (t) at (.25,-1) {$1$} edge [-,semithick] (x);
          \end{scope}
          \begin{scope}[yshift=-1cm,xshift=0cm]
            \node (x) at (0,0) {$x_2$} edge [-,semithick] (d);
            \node (f) at (-.25,-1) {$0$} edge [-,semithick] (x);
            \node (t) at (.25,-1) {$1$} edge [-,semithick] (x);
          \end{scope}
          \begin{scope}[yshift=-1cm,xshift=1cm]
            \node (x) at (0,0) {$x_3$} edge [-,semithick] (d);
            \node (f) at (0,-1) {$0$} edge [-,semithick] (x);
          \end{scope}
        \end{scope}
      \end{scope}
      \begin{scope}[yshift=-1cm,xshift=0cm]
        \node (c) at (0,0) {$c$};
        \path (c) node[above=0.2cm] {${+}$};
        \begin{scope}[yshift=-1cm,xshift=-1.5cm]
          \node (d) at (0,0) {$d$} edge [-,semithick] (c);
          \begin{scope}[yshift=-1cm,xshift=-1cm]
            \node (x) at (0,0) {$x_1$} edge [-,semithick] (d);
            \node (t) at (0,-1) {$1$} edge [-,semithick] (x);
          \end{scope}
          \begin{scope}[yshift=-1cm,xshift=0cm]
            \node (x) at (0,0) {$x_2$} edge [-,semithick] (d);
            \node (f) at (-.25,-1) {$0$} edge [-,semithick] (x);
            \node (t) at (.25,-1) {$1$} edge [-,semithick] (x);
          \end{scope}
          \begin{scope}[yshift=-1cm,xshift=1cm]
            \node (x) at (0,0) {$x_3$} edge [-,semithick] (d);
            \node (f) at (-.25,-1) {$0$} edge [-,semithick] (x);
            \node (t) at (.25,-1) {$1$} edge [-,semithick] (x);
          \end{scope}
        \end{scope}
        \begin{scope}[yshift=-1cm,xshift=1.5cm]
          \node (d) at (0,0) {$d$} edge [-,semithick] (c);
          \begin{scope}[yshift=-1cm,xshift=-1cm]
            \node (x) at (0,0) {$x_1$} edge [-,semithick] (d);
            \node (f) at (-.25,-1) {$0$} edge [-,semithick] (x);
            \node (t) at (.25,-1) {$1$} edge [-,semithick] (x);
          \end{scope}
          \begin{scope}[yshift=-1cm,xshift=0cm]
            \node (x) at (0,0) {$x_2$} edge [-,semithick] (d);
            \node (f) at (0,-1) {$0$} edge [-,semithick] (x);
          \end{scope}
          \begin{scope}[yshift=-1cm,xshift=1cm]
            \node (x) at (0,0) {$x_3$} edge [-,semithick] (d);
            \node (f) at (-.25,-1) {$0$} edge [-,semithick] (x);
            \node (t) at (.25,-1) {$1$} edge [-,semithick] (x);
          \end{scope}
        \end{scope}
      \end{scope}
      \begin{scope}[yshift=-1cm,xshift=8cm]
        \node (c) at (0,0) {$c$} ;
        \path (c) node[above=0.2cm] {${-}$};
        \begin{scope}[yshift=-1cm,xshift=-3cm]
          \node (d) at (0,0) {$d$} edge [-,semithick] (c);
          \begin{scope}[yshift=-1cm,xshift=-1cm]
            \node (x) at (0,0) {$x_1$} edge [-,semithick] (d);
          \end{scope}
          \begin{scope}[yshift=-1cm,xshift=0cm]
            \node (x) at (0,0) {$x_2$} edge [-,semithick] (d);
            \node (f) at (-.25,-1) {$0$} edge [-,semithick] (x);
            \node (t) at (.25,-1) {$1$} edge [-,semithick] (x);
          \end{scope}
          \begin{scope}[yshift=-1cm,xshift=1cm]
            \node (x) at (0,0) {$x_3$} edge [-,semithick] (d);
            \node (f) at (-.25,-1) {$0$} edge [-,semithick] (x);
            \node (t) at (.25,-1) {$1$} edge [-,semithick] (x);
          \end{scope}
        \end{scope}
        \begin{scope}[yshift=-1cm,xshift=0cm]
          \node (d) at (0,0) {$d$} edge [-,semithick] (c);
          \begin{scope}[yshift=-1cm,xshift=-1cm]
            \node (x) at (0,0) {$x_1$} edge [-,semithick] (d);
            \node (f) at (-.25,-1) {$0$} edge [-,semithick] (x);
            \node (t) at (.25,-1) {$1$} edge [-,semithick] (x);
          \end{scope}
          \begin{scope}[yshift=-1cm,xshift=0cm]
            \node (x) at (0,0) {$x_2$} edge [-,semithick] (d);
          \end{scope}
          \begin{scope}[yshift=-1cm,xshift=1cm]
            \node (x) at (0,0) {$x_3$} edge [-,semithick] (d);
            \node (f) at (-.25,-1) {$0$} edge [-,semithick] (x);
            \node (t) at (.25,-1) {$1$} edge [-,semithick] (x);
          \end{scope}
        \end{scope}
        \begin{scope}[yshift=-1cm,xshift=3cm]
          \node (d) at (0,0) {$d$} edge [-,semithick] (c);
          \begin{scope}[yshift=-1cm,xshift=-1cm]
            \node (x) at (0,0) {$x_1$} edge [-,semithick] (d);
            \node (f) at (-.25,-1) {$0$} edge [-,semithick] (x);
            \node (t) at (.25,-1) {$1$} edge [-,semithick] (x);
          \end{scope}
          \begin{scope}[yshift=-1cm,xshift=0cm]
            \node (x) at (0,0) {$x_2$} edge [-,semithick] (d);
            \node (f) at (-.25,-1) {$0$} edge [-,semithick] (x);
            \node (t) at (.25,-1) {$1$} edge [-,semithick] (x);
          \end{scope}
          \begin{scope}[yshift=-1cm,xshift=1cm]
            \node (x) at (0,0) {$x_3$} edge [-,semithick] (d);
          \end{scope}
        \end{scope}
      \end{scope}
    \end{tikzpicture}
    \caption{Reduction of SAT to $\TWIG_0^\pm$-consistency for
      $\varphi_0=(\neg x_1\lor x_2 \lor \neg x_3)\land(x_1\lor \neg
      x_2)$.}
    \label{fig:reduction}
  \end{figure*}

  The building block of the reduction is a {\em brush tree} which is
  used to encode Boolean valuations and constraints on them. For
  instance, for the set of variables $\{x_1,x_2,x_3\}$ the full brush
  tree is $d(x_1(0,1),x_2(0,1),x_3(0,1))$ but typically we remove some
  of the leaves. For instance, the valuation
  $V_0=\{(x_1,\false),(x_2,\false),(x_2,\true)\}$ is represented by
  the tree $t_0=d(x_1(0),x_2(0),x_3(1))$. Note that the tree pattern
  $c(t_0)$ separates the positive examples from the negative ones in
  Figure~\ref{fig:reduction} because $V_0$ satisfies $\varphi_0$.

  The constructed set of examples consists of several $c$-trees. The
  positive $c$-trees specify the satisfying valuations of the input
  $\mathsf{CNF}$ formula; there is one $c$-tree per clause of the
  input formula. Each $c$-tree contains one brush tree per literal of
  the clause, every brush tree encoding the valuations that satisfy
  the corresponding literal (one leaf removed). The negative $c$-tree
  ensures that a brush filter that separates the positive examples
  from negative is well-formed and encodes a valuation. This $c$-tree
  contains one brush tree per variable of the input formula, every
  brush tree has both leaves of the corresponding variable $x_i$
  removed. We claim that this set of examples is consistent if and
  only if the input $\mathsf{CNF}$ formula is satisfiable. The
  \emph{if} part is trivial and the proof of the \emph{only if} part
  is technical and uses the observation that the depth of any twig
  query separating the positive examples from negative ones is bounded
  by $4$.\qed
\end{proof}
The result holds even for very limited query classes that do not use
$\dblslash{}$-edges and $\wc$, and in particular the result hold for
path-subsumption-free twig queries.

The problem of consistency of the input sample in the presence of
positive and negative examples has also been consideed for string
patterns and found to be NP-complete~\cite{MiShSh00}. The proof can be
easily adapted to show the following.
\begin{theorem}
  $\PATH_i^\pm$-consistency is NP-complete for any $i\in\{0,1\}$.
\end{theorem}
We remark, however, that the proof cannot be extended to twig queries
because these are much more expressive even when interpreted over
linear trees (words).

Overall, the negative results for checking consistency give us
\begin{corollary}
  Unless $P=NP$, none of the classes $\Path_i$ and $\Twig_i$ for
  $i\in\{0,1\}$ is learnable in polynomial time and data in the
  presence of positive and negative examples.
\end{corollary}

\section{Related Work}
\label{sec:related-work}
Our research adheres to computational learning theory~\cite{KeVa94}, a
branch of machine learning, and in particular, to the area of language
inference~\cite{Go67}. Our learning framework is inspired by the one
generally used for inference of languages of word and
trees~\cite{OG-inference,OnGa92} (see also \cite{Higuera05} for survey
of the area). Analogous frameworks have been employed in the context
of XML for learning of DTDs and XML
Schemas~\cite{BeNeScVa10,BeGeNeVa10}, XML
transformations~\cite{LeMaNi10}, and $n$-ary automata
queries~\cite{CGLN07}.

Because the positive examples are generally believed to be easier to
obtain, learning from positive examples only is desirable. However,
many classes of languages are learnable only in the presence of both
positive and negative examples e.g., regular languages~\cite{Go67},
deterministic regular languages~\cite{BeGeNeVa10} are not learnable
from positive examples only, in fact any \emph{superfinite} language
class, a class containing all finite languages and at least one
infinite, cannot be learned from positive examples even if we consider
algorithms that do not work in polynomial time. To enable learning
from only positive examples various restrictions have been considered
e.g., reversible languages~\cite{Angluin82}, $k$-testable
languages~\cite{GaVi90}, languages of $k$-occurrence regular
expressions~\cite{BeGeNeVa10}, and $(k,l)$-contextual tree
languages~\cite{RaBrBu08}. What is important to point out here is that
the ability to learn subclasses of path and twig queries from positive
examples comes from the fact that the expressive power of path and
twig queries is relatively weak. Paradoxically, the very same fact is
also responsible for the unfeasability to learn path and twig queries
when both positive and negative examples are present.

Our basic learning algorithm for unary embeddable path queries is
inspired and can be seen as an extension of algorithms for inference
of word patterns~\cite{Angluin79,Shinohara82} (see \cite{ShAr95} for a
survey of the area). A word pattern is a word using extra wildcard
characters. For instance, \emph{regular patterns} use a wildcard
$\oplus$ matching any nonempty string e.g.,
$a\mathord{\oplus}b\mathord{\oplus}c$ matches $aabbc$ and $abbbc$ but
not $abbc$, $abc$, and $cbc$. \emph{Extended regular patterns} use a
wildcard $\circledast$ that matches any (possibly empty) string e.g.,
$a\mathord{\circledast}b$ matches $ab$ and $acbcb$. To capture unary
path queries we need to use the wildcard $\oplus$ and another wildcard
$\odot$ that matches a single letter, and then for instance the
pattern $a\mathord{\oplus}b\mathord{\odot}c$ corresponds to the path
query $/a\dblslash{}b/\wc/c$ when interpreted over paths of the input
tree. We observe that $\oplus$ is equivalent to
$\mathord{\circledast}\mathord{\odot}$ and engineer our learning
algorithm using the ideas behind the algorithms for inference of
regular patters~\cite{Angluin79} and extended regular
patters~\cite{Shinohara82}.

Learning of unary XML queries has been pursued with the use of node
selecting tree automata~\cite{CGLN07}, with extensions allowing to
infer $n$-ary queries~\cite{LeNiGi06}, take advantage of schema
information~\cite{CGLN08}, and use pruning techniques to handle
incompletely annotated documents~\cite{CGLN07}.
The main advantage of using node selecting tree automata is their
expressive power. Node selecting tree automata capture exactly the
class of $n$-ary MSO tree queries~\cite{ThWr68,LeNiGi06}, which
properly includes twig and path queries. However, tree automata have
several drawbacks which may render them unsuitable for learning in
certain scenarios: this is a heavy querying formalism with little
support from the existing infrastructure and it does not allow an easy
visualization of the inferred query. 

Although, the class of twig queries is properly included by the class
of MSO queries and path queries are captured by regular languages,
using automata-based techniques to infer the query and then convert it
to twigs is unlikely to be successful because automata translation is
a notoriously difficult task and typically leads to significant
blowup~\cite{EhZe76} and it is generally considered beneficial to
avoid it~\cite{Fernau04}. An alternative approach, along the lines
of~~\cite{BeNeScVa10}, would be to define a set of structural
restrictions on the automaton that would ensure an easy translation to
twig queries and enforce those conditions during inference. However,
such restrictions would need to be very strong, at least for twig
queries, and this approach would require significant modification of
the inference algorithm, to the point where it would constitute a new
algorithm.

Methods used for inference of languages represented by automata differ
from the methods used in our learning algorithms. An automata-based
inference typically begins by constructing an automaton recognizing
exactly the set of positive examples, which is then generalized by a
series of generalization operation e.g., fusions of pairs of
states. To avoid overgeneralization of the automata, negative examples
are used to filter only consistent generalizations
operations~\cite{OnGa91}, and if negative examples are not available,
structural properties of the automata class can be used to pilot the
generalization process~\cite{Angluin82,GaVi90,BeGeNeVa10}. Our
algorithms, similarly to word pattern inference
algorithms~\cite{Angluin79,Shinohara82}, begin with the universal
query and iteratively specialize the query by incorporating
subfragments common to all positive examples.

XLearner~\cite{MoKiMa04} is a practical system that infers XQuery
programs. It uses Angluin’s DFA inference algorithm~\cite{Angluin87a}
to construct the XPath components of the XQuery program. The system
uses direct user interaction, essentially equivalence and membership
queries, to refine the inferred query. Because of that the learning
framework, called the \emph{minimally adequate
  teacher}~\cite{Angluin87a}, is different from ours and allows to
infer more powerful queries. We also point out that learning twigs is
not feasible with equivalence queries only~\cite{CaCeGo06}.

Raeymaekers et al. propose learning of $(k,l)$-contextual tree
languages to infer queries for web
wrappers~\cite{RaBrBu08}. $(k,l)$-contextual tree languages form a
subclass of regular tree languages that allows to specify conditions
on the nodes of the tree at depth up to $l$ and each condition
involves exactly $k$ subsequent children of a node. Because only nodes
at bounded depth can be inspected and the relative order among
children is used, $(k,l)$-contextual tree languages are incomparable
with twig queries which can inspect nodes at arbitrary depths but
ignore the relative order of nodes.

Finally, we point out that the problem of query inference has been
studied in the setting of relational
setting~\cite{SPGW10,TrChPa09,GiBu09}. Relational databases and their
query languages offer a set of opportunities and challenges radically
different from those encountered in semi-structured databases. For
instance, the query inference involves constructing a desired
selection condition that yields the required tuples from a table, a
task that easily becomes intractable. 

\section{Conclusions and future work}
\label{sec:concl-future-work}
We have studied the problem of inferring an XML query from examples
given by the user. We have investigated several classes of Boolean and
unary, path and twig queries and considered two settings for the
problem: one allowing positive examples only and one that allows both
positive and negative examples. For the setting with positive examples
only, we have presented sound and complete learning algorithms for
practical subclasses of queries: anchored path queries and
path-subsumption-free twig queries. On the other hand, inclusion of
negative examples to the input sample renders learning unfeasible.

We believe that negative examples have an important informative
quality and we intend to investigate approaches that take advantage of
it. Two directions are possible: relaxing the definition of
learnability and extending the query class. A notion allowing the
query to select some negative examples and omit some positive examples
is a natural direction of making our learning algorithms capable of
producing queries of better quality (cf.\ Example~\ref{ex:2}) and able
to handle noisy samples. For the second direction, our preliminary
results show that adding union to the query languages renders
consistency quite simple to decide but the satisfaction of \Prop1 and
\Prop2 is not clear, and therefore, new learning techniques need to be
developed. We are also interested in extending the query language with
other operators (e.g., negation) and see their impact on learnability.

We observe that the main reason for restricting our attention to
anchored path queries are the properties \Prop1 and \Prop2 defined in
Section \ref{sec:classes} that allow to use embeddings to equate the
semantics of the query with its structure and enforce the existence of
match sets of polynomial size. \cite{MiSu99,MiSu04} introduced adorned
path queries, allowing to represent $\dblslash{}\wc/\wc$ as
$\dblslash{}^{\geq2}$, and extended embeddings to homomorphisms of
adorned queries. Homomorphisms are shown to connect tightly the
structure of path queries and their semantics (\Prop1). It would be
interesting to see to what extent the notion of homomorphism could be
used to improve learnability results.  We point out that for path
queries the only know construction of match sets produces exponential
sets. Moreover, the homomorphism technique does not work for twig
queries.

Finally, we would like to enable our algorithms to take advantage of
schema information (cf.\ Example~\ref{ex:learner-psf-twig1}). The
schema may be given explicitly e.g., as a DTD, or implicitly as a
result of a learning algorithm. Because testing the containment of
XPath queries in the presence of DTDs is know to be intractable in
general~\cite{DeTa01,NeSc03,BeFaGe05} and in fact most of the
reductions showing hardness use (or can be modified to use) anchored
queries, the use of DTDs in this context may be quite limited and we
intend to investigate alternative schema formalisms tailored for query
learning.


\end{document}